\documentclass[5p,times,authoryear]{elsarticle}

\usepackage{amssymb}
\usepackage{amsmath}
\usepackage{algpseudocode}
\usepackage{algorithm}
\usepackage{mathtools}

\DeclarePairedDelimiter{\floor}{\lfloor}{\rfloor}

\usepackage[export]{adjustbox}
\usepackage{hyperref}
\algnewcommand\algorithmicforeach{\textbf{for each}}
\algdef{S}[FOR]{ForEach}[1]{\algorithmicforeach\ #1\ \algorithmicdo\ }
\usepackage{amsthm}
\newtheorem{theorem}{Theorem}

\newcommand{\argmin}{\operatornamewithlimits{argmin}}
\newcommand{\argmax}{\operatornamewithlimits{argmax}}

\journal{Astronomy and Computing}

\begin{document}

\begin{frontmatter}



\title{Unsupervised Learning of Structure in Spectroscopic Cubes}


\author[inst1]{Mauricio Araya\corref{cor1}}
\author[inst1]{Marcelo Mendoza}
\author[inst1]{Mauricio Solar}
\author[inst2]{Diego Mardones}
\author[inst3]{Amelia Bayo}
\address[inst1]{Departamento de Inform\'atica, Universidad T\'ecnica Federico Santa Mar\'ia, Casilla 110-V, Valpara\'iso, Chile}
\address[inst2]{Departamento de Astronom\'ia, Universidad de Chile, Casilla 36-D, Santiago, Chile}
\address[inst3]{Instituto de F\'isica y Astronom\'ia, Universidad de Valpara\'iso, Casilla 5030, Valpara\'iso , Chile}
\cortext[cor1]{mauricio.araya@usm.cl}

\begin{abstract}
We consider the problem of analyzing the structure of spectroscopic cubes 
using unsupervised machine learning techniques. We propose representing the target's signal as an
homogeneous set of volumes through an iterative algorithm that separates the 
structured emission from the background while not overestimating the flux. 
Besides verifying some basic theoretical properties, the algorithm is designed to 
be tuned by domain experts, because its parameters have meaningful values in the 
astronomical context. Nevertheless, we propose an heuristic to automatically estimate
the signal-to-noise ratio parameter of the algorithm directly from data. 
The resulting light-weighted set of samples ($\leq 1\%$ compared to the original data) offer several advantages. For instance, 
it is statistically correct and computationally inexpensive to apply well-established techniques of the pattern 
recognition and machine learning domains; such as clustering and dimensionality reduction algorithms. 
We use ALMA science verification data to validate our method, and present examples of the operations 
that can be performed by using the proposed representation. Even though this approach is focused on providing faster and 
better analysis tools for the end-user astronomer, it also opens the possibility of 
content-aware data discovery by applying our algorithm to big data.

\end{abstract}


\end{frontmatter}


\section{Introduction}

Even though there is a large body of work regarding 2D image ana\-ly\-sis, 
most of the techniques do not directly scale up to more dimensions.
There are specific-purpose algorithms in astronomy to deal with 3D data,
such as clump finding algorithms for spectroscopic data cubes 
\citep{williams94,stutzki90,berry15},
yet the current state of the practice requires a huge effort
in terms of storage space, computational time and actually human-machine
interaction to generate useful products for astronomers \citep{mcmullin07}. 
Moreover, data growth in sensitivity and resolution with each new instrument, 
so the next-generation of projects in Astronomy will produce
several terabytes of data every night \citep{ivezic08,dewdney08}, making 
impossible to perform analysis without automatically reducing its dimensionality. 

Machine learning, and other advanced statistical methods, have
been a source of success for astronomers \citep{astroML,richards11,gibson12}:
learning and inference are powerful
tools to represent data in a compact way that allows us to make automatic
decisions. Machine learning methods for classification, model-based regression, clustering and feature selection \citep{bishop06}, often rely on samples being independent and identically distributed
(i.i.d.), which is not the case for the pixels of an image. Therefore, applying
state-of-the-art machine learning techniques usually involves adapting the
method or preprocessing the data to comply with this assumption.

We propose representing spectroscopic data cubes as a compact and homogeneous set of
volumes, that can be treated directly as samples of the underlying 
signal of interest to achieve two goals:

\begin{itemize}
\item reduce the size of the cube representation to limit computational
and memory resources needed to perform astronomical analysis, and
\item comply with the i.i.d. assumption allowing astronomers 
to use machine learning and statistical analysis tools that are based on this assumption.
\end{itemize} 

In this article we focus on spectroscopic data cubes, and specifically on interferometric synthesized spectral cubes on the millimeter/sub-millimeter range, but this work could be straightforwardly applied to 2D images. Moreover, new acquisition techniques are now producing even higher dimensional data with axes
such as polarization, spatial depth or time series. For instance, the data that is
currently generated by the ALMA observatory \citep{testi10} are actually
4D data hypercubes. Again, our algorithm can work with these higher dimensions without major changes.

The outline of the article is the following. 
In Section~\ref{sec:soa} we describe the problem of analyzing
spectroscopic cubes by reviewing current approaches to do this.
Section~\ref{sec:hcr} presents an homogeneous compact representation for spectroscopic cubes, while Section~\ref{sec:emp} presents the experimental results of computing this representation. In Section~\ref{sec:apps} we show how to use our representation for data analysis. Finally, we conclude in Section~\ref{sec:conc} giving remarks and discussing future work.

\section{From 2D to 3D and Beyond}
\label{sec:soa}

The high dimensionality and high sensitivity of new detectors and instruments 
arise new problems for storage, processing and transferring astronomical data. 
The \emph{Big Data} problem in astronomy is not only about managing a large number of 
observations, but also to deal with files of large size.

The research on 2D image analysis techniques, such as
segmentation~\citep{russ06} and denoising~\citep{motwani04}, has been largely addressed using the standard
pixel-based representation. The use of these techniques is widely
spread in astronomy \citep{starck02}, including specialized packages that
directly produce catalogs from observations \citep{Bertin96}.

However, extending these algorithms to more
dimensions is a non-trivial task. The main problem is the curse
of dimensionality \citep{donoho00}: not only the amount of data growths
exponentially in terms of bytes with more dimensions, but the algorithms used for 2D images become
intrinsically more complex. 
Between the difficulties for extending 2D algorithms to more dimensions we want to highlight:

\begin{itemize}
\item \emph{Resulting Images}: pixel-based representations can give us high
level information in 2D images, such as colored segments, thresholded views
and contour images. However they become hard to visualize, annotate,  store
and transfer in more dimensions \citep{kozak15}. A clear example is pixel-based
segmentation \citep{russ06}, which assigns a
class to each (significant) pixel. Unfortunately, in spectroscopic cubes a pixel-based
segmentation\footnote{We use pixels
as the generic name for elements of a n-dimensional array, rather than voxels
or other dimensionality-dependent names.}
result is almost as hard to store, transport and analyze than the original
cube, because users can only work with 2D projections of the result.
Additionally, merging different cubes requires extra steps like rotating,
down-sampling and interpolating data, which could affect the content of the
cube.
\item \emph{Pixel Vicinity}: the complexity of the vicinity of pixels
strongly depends on the dimensionality. For
example, a connectivity by contact in a 2D image consist only of 8 pixels (4
edges and 4 vertices), but
this becomes 26 pixels for 3D cubes (6 squares, 12 edges and 8 vertices) and
80 pixels for 4D cubes (8 cubes, 24 squares, 32 edges and 16 vertices).
Complex statistical models, such as Markov random fields \citep{kato12}, will not only
increase its computational time due to augmented number of pixels, but also by the
non-linear increase of pixel connectivity. Moreover, pixel vicinity
becomes a hard problem when two cubes with different resolutions need to
be merged.
\item \emph{Intensity Dilution}: the intensity of a phenomenon observed in more
dimensions dilutes, so some structures can be detected only statistically due
to the low signal-to-noise ratio. A clear example of this problem arises in
thresholding techniques, where large structures that can be easily detected in
low dimensional projections of the cube might be neglected because
most of their components fall below the applied threshold. 
\end{itemize}

To tackle these problems, astronomers perform data preprocessing operations such as
binning, resampling and integrating data, which allow them to analyze lower-dimensional
projections of data. Also, astronomers use content-based operations to such 
as background filtering, pixel masking, or automatic region-of-interest 
detection, in order to analyze more compact representations of data. 
For example, a very simple detection method is thresholding, which allow us
to select those pixels that are over some flux value (e.g., the RMS) and discard
the rest in the same way that is done for 2D images. This can be combined with 
a low-pass filter to smooth the signal and select regions of interest rather than isolated pixels. 
More advanced methods use morphological transformations (structured elements or kernel-density functions) and edge detection techniques such as in \citet{mendoza16}. 

An interesting family of detection methods is pixel-based clumping algorithms
for spectroscopic data cubes, which separate the signal not only from background, but 
\emph{clusterize} pixels in different emission sources. The clumpy structure of molecular clouds and
other extended objects in astronomy, allows separating sources to analyze them
independently or to find the astrophysical relationships between them. 
The most known method in this category is ClumpFind \citep{williams94}, which uses 
contours at different RMS levels to define clumps.
A relatively newer method is FellWalker \citep{berry15}, that uses hill-climbing and 
cellular automata techniques to offer a more intelligent separation.
These methods select those pixels (or regions) that are unlikely explained
by the background. However, it is difficult to determine how much flux of each pixel can be
explained by the signal or by the background. 

Gaussclumps \citep{stutzki90} is an almost 30 years old method that iteratively fit Gaussian
structures directly in both spatial and spectral dimensions. The method is very
well motivated in the sense that emissions and clump structures can often be
accurately represented using Gaussians. 
The fitting is done using non-linear optimization under several soft constraints 
that emerge from astrophysical restrictions, based on a simple radiative emission 
model \citep{stahler2008}.
Unfortunately, the algorithm is complex to analyze from an algorithmic
perspective and difficult to use in practice.  First, the algorithm is composed by several iterative heuristics and optimization steps, each one with multi-criteria halting conditions that depend of free parameters. This prevents a proper analysis of the convergence and computational complexity of the algorithm, because the number of variables to analyze are too many. Moreover,
most of the parameters are meaningless from the astrophysical point
of view, making them hard to tune. For example, the number of 
consecutive iterations where the non-linear optimizer failed to converge, is a
halting parameter that is unlikely to be tuned \emph{a priori}, and its
connection to astrophysical parameters is very weak or non-existent. 
In addition, there are more than thirty parameters in the
off-the-shelf implementation \citep{berry2007cupid}. Thus, 
performing a thorough sensitivity analysis becomes impractical.

However, is important here to distinguish between model and method: 
while the model is
well motivated and sound, the algorithm used for finding the Gaussian components
does not offer theoretical guarantees and is hard to tune.

\vspace{-2mm}

\section{Homogeneous Compact Representation}
\label{sec:hcr}

\vspace{-1mm}

Our proposal is a simple yet powerful approach to spectral cube analysis,
which consists in decomposing the data in several identical volumes that represent
the actual signal of the cube, similar to kernel density estimation. 
As target signals are relatively sparse in a 3D space, this method will produce a compact
representation with homogeneous components that are 
easier to fit in memory and to process compared to the original data. 

For the sake of generality, we define that an cube $C$ as a n-dimensional matrix of scalars that can be decomposed by linear combination of kernel instances $\hat{C}$ plus a noise (Equation~\ref{eq:model}).

\begin{equation}
C(\mathbf{x}) \simeq \hat{C}(\mathbf{x};\theta) + \epsilon = \sum_i^k
\alpha_i K(\mathbf{x},\mathbf{y}_i) + \epsilon ,
\label{eq:model}
\end{equation}

\noindent where $\mathbf{x} \in \mathcal{X} $ is a coordinate of the cube, $\mathcal{X} \subset \mathbb{N}^d$ is the set of all pixels coordinates of the cube, 
$K(\cdot,\cdot) \mapsto [0,1]$ is a kernel function, $\mathbf{y}_i \in \mathbb{R}^d$ are
kernel location points, $\alpha_i$ are positive scalars and
$\epsilon$ is a random-noise variable. We denote by $\theta$ the set of model parameters, in this case $\mathbf{y}_i$ and $\alpha_i$. In general, the kernel function could
be replaced with any function with arbitrary parameters that can be added to
$\theta$, for instance by a structured Gaussian function like in \citet{stutzki90}. 
However, as our objective is to produce a compact representation, we constraint the functions to kernels with location points in the same $\mathcal{X}$ space
\footnote{Formally, both variables are not in the same space because $\mathbf{y}$ lives in a continuous space while $\mathcal{X}$ is discrete, but both represent positions.}.
To simplify even more the model we propose choosing constant intensities in
the linear combination, i.e. $\alpha_i = \sigma, \forall i$. 
Accordingly, our representation is \emph{homogeneous}, because each volume has the same structure (kernel) and 
contains the same energy (constant intensity). Also, the kernel should comply with $K(x,x) = 1$ for convenience.

We assume that $\epsilon$ is Gaussian, but we constraint the noise only to positive values because the background noise 
is additive in astronomy. This constraint can be modeled by letting $\epsilon \sim f(\sigma)$ be 
half-normally distributed --- a zero-mean normal distribution with positive values only --- where the $\sigma$ parameter
can be estimated from data (e.g., computing the RMS). Please note that this \emph{positivity constraint} is only applied to the noise model, allowing the cube to have negative values due to flux calibration or continuum subtraction. Under these assumptions, the generic emission detection problem can be casted as
the maximization of the log-likelihood function:

\begin{align}
\argmax_\theta
\left[\ln\left(\mathcal{L}(C|\theta,\sigma)\right)\right] &=
\argmax_\theta\left[  
\sum_{\mathbf{x} \in \mathcal{X}} \ln \left(f(C(\mathbf{x}) -  \hat{C}(\mathbf{x};\theta) ;
\sigma)\right)\right] \nonumber \\
&= \argmin_{\mathbf{y}_1,\ldots,\mathbf{y}_k}\left[  
\sum_{\mathbf{x} \in \mathcal{X}} \left(C(\mathbf{x}) -
\sigma \sum_i^k K(\mathbf{x},\mathbf{y}_i) \right)^2 \right], 
\label{eq:likelihood}
\\
 s.t. \hspace{5mm} & 
\sigma \sum_i^k K(\mathbf{x},\mathbf{y}_i) \leq C(\mathbf{x}) , \forall\mathbf{x} \in \mathcal{X} ,
\label{eq:constraint}
\end{align}

In summary, the problem consist in 
solving a least-squared problem (Equation~\ref{eq:likelihood}) with a positivity constraint (Equation~\ref{eq:constraint}).

If the cube is sparse with respect to the noise level, we expect to produce 
significantly less parameters in $\theta$ than the number of original pixels in $\mathcal{X}$, 
keeping in the representation enough information to perform data analysis. 

For example, consider that we use only the pixels above a flux level $\tau$, 
and choose $K(\mathbf{x},\mathbf{y}) = 1$ if $x = y$ and $0$ elsewhere. 
Then, Equation~\ref{eq:likelihood} can be solved by representing the cube as a
set of homomorphic box volumes, in which each pixel $i$ above $\tau$ is represented by $n_i$ boxes.
Formally, Equation~\ref{eq:likelihood} reduces to

\begin{align*}
\argmax_\theta
\left[\ln\left(\mathcal{L}(C|\theta,\sigma)\right)\right] 
&=\argmin_{\mathbf{y}_1,\ldots,\mathbf{y}_k}\left[  
\sum_{\mathbf{x} \in \mathcal{X}} \left(C(\mathbf{x}) -
\sigma \sum_i^k I(\mathbf{x},\mathbf{y}_i) \right)^2 \right] \\
&=\argmin_{\mathbf{n} \in \mathbb{N}^k}\left[  
\sum_{\mathbf{x} \in \mathcal{X}} \left(C(\mathbf{x}) -
\sigma n_\mathbf{x} \right)^2 \right], \\
s.t. \hspace{5mm} & \sigma n_\mathbf{x} \leq C(\mathbf{x}), \forall \sigma n_\mathbf{x} > \tau,
\end{align*}

\noindent and its optimal solution is 

\begin{equation*}
\theta = \{n_\mathbf{x} = \left\lfloor C(\mathbf{x})/\sigma \right\rfloor \textrm{ } |  \textrm{ } C(\mathbf{x}) >
\tau , \mathbf{x} \in \mathcal{X}\},
\end{equation*}
where $\floor{\cdot}$ indicates the floor function (i.e. integer part of the real expression).

This pixel-based homogeneous representation already reduces the size of the representation with almost no computational effort (refer to Section \ref{sec:phr} for more details). 
However, smaller and better representations can be found with the following algorithm.

\subsection{Iterative Bubble Subtraction Algorithm}

For spread kernels, the squared difference of Equation~\ref{eq:likelihood} could be 
solved using non-linear constrained optimization. However, as $k$ depends on the emission structure, signal-to-noise ratio (SNR), and intensity,
among other factors, its value is unknown. Moreover, as the problem is non-convex, the use of numerical solvers
is computationally expensive. Even by constraining the collection of location points to $\mathbf{y}_i \in \mathcal{X}$ we end up 
with a hard combinatorial optimization problem. 

We propose using a simple iterative algorithm that subtract at each step 
the volume formed by the kernel within a forced compact support. The ``bubble'' is subtracted from 
the position that holds the \emph{maximum safe energy} with respect to the kernel, constraining the substraction to a fixed intensity $\sigma$. This leads to a fast algorithm that produce an homogeneous representation of the signal. The main feature of the representation is that each component is compactly represented only by its location
$\mathbf{y}_i \in \mathcal{X}$. 

The \emph{safe energy} at each point is the maximum intensity of the bubble that can be subtracted from the cube without 
surpassing the actual flux of the pixels. 
Let $R$ be the residual cube (i.e., a working copy of the original $C$ cube),
$K$ the kernel and $\mathbb{K}_\mathbf{x}$ the set of pixels of the compact support of the kernel
around $\mathbf{x}$. Then, the safe energy cube $E$ is computed by dividing the residual by each element of the bubble (Equation~\ref{eq:residual}).

\begin{equation}
E(\mathbf{x}) = \min_{\mathbf{p} \in \mathbb{K}_{\mathbf{x}}}
\left[\frac{R(\mathbf{p})}{K(\mathbf{p},\mathbf{x})} \right], \forall
\mathbf{x} \in \mathcal{X}.
\label{eq:residual}
\end{equation}

\begin{algorithm}
\small
\caption{Update Energy}
\label{alg:update_energy}
\begin{algorithmic}[1]
   \Function{Update-Energy}{$R,E,\mathbf{y}$}
   \If {$\mathbf{y} == \textrm{None}$} \hfill \# if no center is given,
      \State $\mathcal{S} \gets \mathcal{X}$ \hfill \# compute for the whole
cube
   \Else
      \State $\mathcal{S} \gets \mathbb{K}_\mathbf{y}$ \hfill  \# update only the
modified section
   \EndIf
   \ForEach {$\mathbf{x} \in \mathcal{S}$}
       \State $E(\mathbf{x})=\min_{\mathbf{p} \in \mathbb{K}_\mathbf{x}} 
\frac{R(\mathbf{p})}{K(\mathbf{p},\mathbf{x})}$ \hfill \# safe energy 
   \EndFor
   \EndFunction
\end{algorithmic}
\end{algorithm}
We present here a high-level version of the \textsc{Update-Energy} function
(see Algorithm~\ref{alg:update_energy}) that is simple to grasp yet inefficient in
practice. To speed-up calculations, the bubble can be pre-computed, 
properties of the kernel can be exploited, and fast matrix operations can be
used instead of iterations. In fact, to compute the initial energy of a 3D cube we actually 
divide the elements of $R$ by each element of an eighth of the pre-computed bubble (symmetrical kernel), 
mapping the correct indices to find the energy at each coordinate. In addition, please note that updating $E$ after a
subtraction with forced compact support is very fast, because we need to compute only the neighbor elements in the $\mathbb{K}_\mathbf{x}$ set. 

\begin{algorithm}
\small
\caption{Bubble-Detect}
\label{alg:bubble-detect}
\begin{algorithmic}[1]
   \Function{Bubble-Detect}{$C,K,\sigma,\tau$}
   \State $Y \gets \{\}$
   \State $R \gets C$ \hfill \# working copy 
   \State $\mathbf{y} \gets \textrm{None}$
   \Repeat
      \State \textsc{Update-Energy}($R,E,\mathbf{y}$)
      \State $\mathbf{y} \gets \argmax_\mathbf{x \in \mathcal{X}}
\left[E(\mathbf{x})\right]$ \hfill \# maximum safe energy 
      \If {$E(\mathbf{y}) \leq \sigma$} \hfill \# exit if the maximum energy is lower
than the noise level
         \State \textbf{break}
      \EndIf
      \ForEach {$\mathbf{x} \in \mathbb{K}_\mathbf{y}$} \hfill  \# subtract a bubble from the residual
         \State $R(\mathbf{x}) \gets R(\mathbf{x}) - a\cdot
K(\mathbf{x},\mathbf{y})$ 
      \EndFor
      \State $Y \gets Y \cup \{\mathbf{y}\}$ \hfill \# save the bubble position
   \Until{$E(\mathbf{y}) \leq \tau$} \hfill \# stop when the energy is less
than a threshold
   \State \Return $Y$ 
   \EndFunction
\end{algorithmic}
\end{algorithm}

The \textsc{Bubble-Detect} function (see Algorithm~\ref{alg:bubble-detect}) 
subtracts bubbles until the maximum safe energy falls below a $\tau$ threshold. 
It returns a set of pixel locations $Y$, each one represented by a single integer value 
if the correct encoding is chosen. 

The parameters of the algorithm are designed to be meaningful for astronomers. 
%
The $\sigma$ parameter corresponds to the noise level of the observation. In
our implementation this is equal to the empirical RMS, but the parameter could be
overwritten by the astronomer with a better estimation, for example when the source 
is as extended as to dominate the signal of the whole cube.
The $\tau$ parameter is the target minimum signal that the algorithm
should consider as a detection. This means that $\frac{\tau}{\sigma}$ is the
signal-to-noise ratio frontier where the signal is indistinguishable from noise. 
This also strongly depends on the astronomer knowledge about the nature
of the signal and the noise. A $\tau=\sigma$ means that we stop subtracting
bubbles when a bubble intensity reaches the noise limit. In our experiments
we introduce an heuristic to estimate $\tau$.

Our algorithm verify some basic properties regardless the kernel used.

\begin{theorem}[Upper bounded iterations]
\label{th:1}
Let $F_\tau$ and $n_\tau$ be both the integrated flux and the number of pixels with intensity greater than 
$\tau$ respectively. The number of steps (and solution size) $n_\textsc{BD}$ of the \textsc{Bubble-Detect} algorithm 
is finite, deterministic and upper bounded by $\frac{F_\tau - n_\tau(\tau - \sigma)}{\sigma}$.
\end{theorem}

\begin{proof}
The algorithm is deterministic because the internal state (i.e. residual and energy) is only modified by constants and
pre-computed values (bubble). Please be aware that here we have made the reasonable assumption that maximum and minimum operators are deterministic.
The values of the residual and the energy matrix monotonically decreases at each step, so the sequence of bubble subtraction is finite. 

The solution $Y$ of the algorithm can be represented by a sparse vector of integers $\mathbf{n}$, where $n_\mathbf{x}$ corresponds to the number of times that a
bubble was subtracted at the index $\mathbf{x}$. 
Note that if $C(\mathbf{x}) < \tau$ the maximum possible energy at that point is

\begin{equation}
E(\mathbf{x}) \leq 
\frac{R(\mathbf{x})}{K(\mathbf{x},\mathbf{x})} \leq C(\mathbf{x}) < \tau,
\end{equation}

\noindent implying that $n_\mathbf{x} = 0$.
For all the other $n_\mathbf{x}$, the algorithm can subtract at most the
number of identity kernels under the same threshold, which correspond to
$\left\lfloor \frac{C(\mathbf{x}) - \tau}{\sigma} \right\rfloor$ bubbles from $\mathbf{x}$, because
nearby kernels can only reduce this number, $k(\mathbf{x},\mathbf{x})=1$, and the 
$E(\mathbf{x}) \leq C(\mathbf{x}) - \hat{C}(\mathbf{x})$.

Then, the number of elements
of the solution is bounded by
\begin{align}
n_{\textsc{BD}} & = \sum_{\mathbf{x} \in \mathcal{X}} n_{\mathbf{x}}
 \leq \sum_{\mathbf{x} \in \{\mathbf{x} \in \mathcal{X} | C(x) \geq \tau\} }
\left\lfloor \frac{C(\mathbf{x}) - \tau}{\sigma}\right\rfloor  \nonumber \\
& \leq  \sum_{\mathbf{x} \in \{\mathbf{x} \in \mathcal{X} | C(x) \geq \tau\} }
\frac{C(\mathbf{x}) - \tau + \sigma}{\sigma}  = \frac{F_\tau - n_\tau(\tau - \sigma)}{\sigma} 
\end{align}

\end{proof}

\begin{table*}[tbp]
\normalsize
\centering
\begin{tabular}{l|cccccccc}
 Name & RA & DEC & FREQ & Total & Valid & ARes & BSize & SRes \\
 & $\mathrm{pix}$ & $\mathrm{pix}$ & $\mathrm{pix}$ & $\mathrm{pix}$  &  &
$\mathrm{{}^{\prime\prime}}$ & $\mathrm{{}^{\prime\prime}}$ & $\mathrm{MHz}$ \\
\hline
Orion-KL-CH3OH & 100 & 100 & 41 & 410000 & 100\% & 0.40 & 1.38 &
0.49 \\
TW-Hya-CO(3-2) & 100 & 100 & 118 & 1180000 & 100\% & 0.30 & 1.53 &
0.14 \\
M100-CO(1-0) & 600 & 600 & 40 & 14400000 & 58\% & 0.50 &
2.48 & 3.90 \\
IRAS16293-220GHz & 220 & 220 & 480 & 23232000 & 54\% & 0.20 &
0.99 & 0.49 \\
\end{tabular}
\caption{Summary of the data used in our experiments. The first column describes
the name of the source and the observed spectral line. The next four columns
summarize the dimensions of the cubes. The column denoted by ``valid'' shows how many pixels
have actual flux. The last three columns show angular resolution (ARes), minor beam size (BSize) and 
spectral resolution (SRes).}
\label{tab:dataset}
\end{table*}
\normalsize

\subsection{Diagonal Gaussian Kernel}

Selecting the right kernel is an important issue for this representation because accuracy, 
compactness and computational time are heavily affected by the type of kernel chosen. 
More complex kernels usually require to be smooth functions, because resolved astronomical sources,
such as clouds of dust and gas, usually decay smoothly. A natural choice is to use Gaussian functions, 
because both spatial densities and spectral lines can be efficiently represented by Gaussian mixtures \citep{stutzki90}. 
Gaussians could potentially have any shape, but we recommend
using a diagonal co-variance matrix that captures the resolution of the detector or
instrument used for acquiring the data.
In addition, Gaussians have no compact support. We propose forcing the compact support of a Gaussian approximation kernel, 
by using the resolution information from the metadata of the cube (e.g., spectral resolution, beam size, LAS, etc.). 

Formally, let $\Delta$ be a resolution vector containing the number of pixels
from the center that defines the compact support for each dimension. Then,
a Gaussian kernel co-variance matrix
$\Sigma$ that has a contour of $\gamma$ at the resolution
boundary can be computed as:

\begin{align}
\Sigma &=  \frac{1}{-2 \log(\gamma)} \left[ \begin{array}{cccc}
\Delta_1^2 & 0 & \dots & 0 \\
0 & \Delta_2^2 & \dots & 0 \\
\vdots & \vdots & \ddots & \vdots \\
0 & 0 & \dots & \Delta_n^2 \end{array} \right],
\end{align}

\noindent where the Gaussian approximation kernel correspond to:

\begin{align}
K(\mathbf{x},\mathbf{y}) & = \exp\{-0.5 (\mathbf{x} - \mathbf{y})^\top
\Sigma^{-1}
(\mathbf{x} - \mathbf{y})\}, \forall \mathbf{x} \in \mathbb{K}_{\mathbf{y}},
\end{align}

\noindent and $\mathbb{K}_{\mathbf{y}} = \{ \mathbf{x} \in \mathcal{X} \mid
\mathbf{x} \in [\mathbf{y} - \Delta,\mathbf{y} + \Delta] \}$.
$\Delta$ defines the size of the bubble in pixels per dimension, 
and the minimum values of it should depend on the instrumentation or physical limitations
(e.g., beam size, point spread function or minimum expected broadening).
Larger values will produce a less precise approximation with less elements,
but not forcedly in less time due to energy computations. In other words, 
$\Delta$ controls the granularity of the approximation.

$\gamma$ controls the contour level of the Gaussian kernel,
which can be understood as the degree of ``smoothness'' that we want in the solution.
$\gamma$ should comply with $0 < \gamma < 1$, where $0$ means maximal smoothing and $1$ means maximal sharpness. 
Please note that if the kernel want to be used as continuous approximation of the data, then the values of
$\gamma$ should be small. By default our implementation uses $\gamma=0.1$.

\section{Empirical Evaluation}
\label{sec:emp}

\begin{figure}[tbp]
  \begin{center}
   \includegraphics[width=0.99\linewidth]{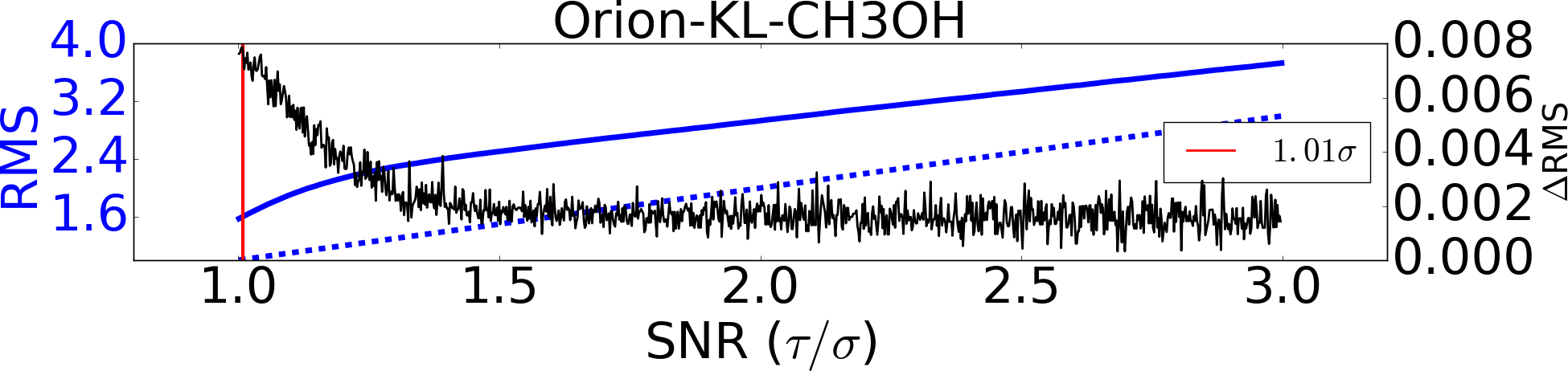}
   \vspace{2mm}
     \includegraphics[width=0.99\linewidth]{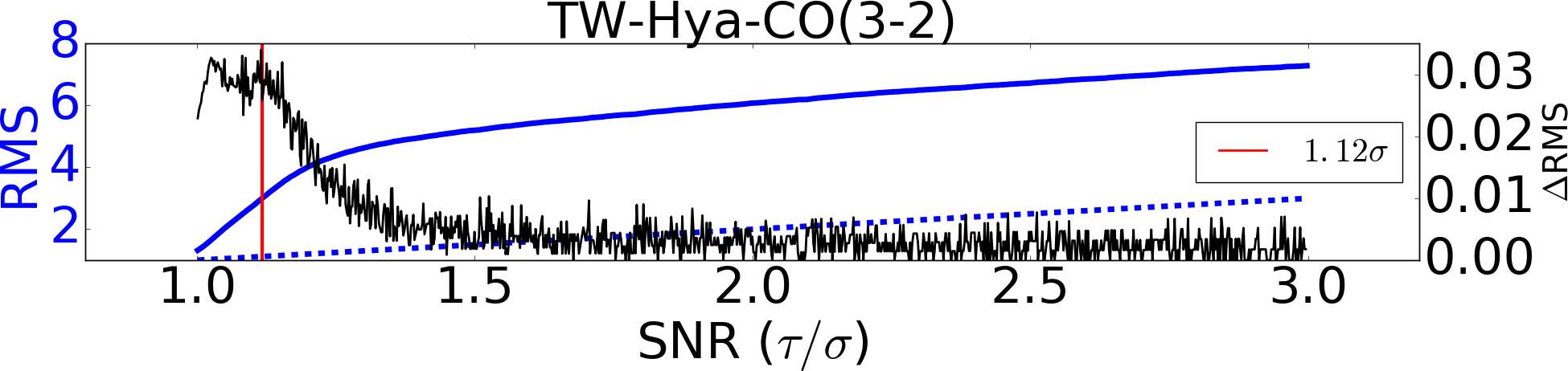}
     \vspace{2mm}
       \includegraphics[width=0.99\linewidth]{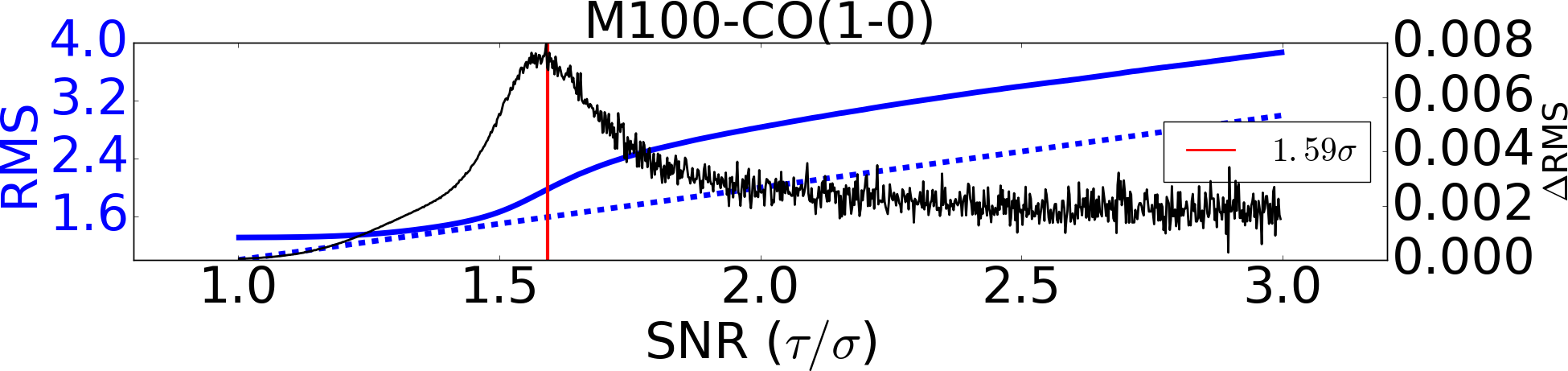}
       \vspace{2mm}
         \includegraphics[width=0.99\linewidth]{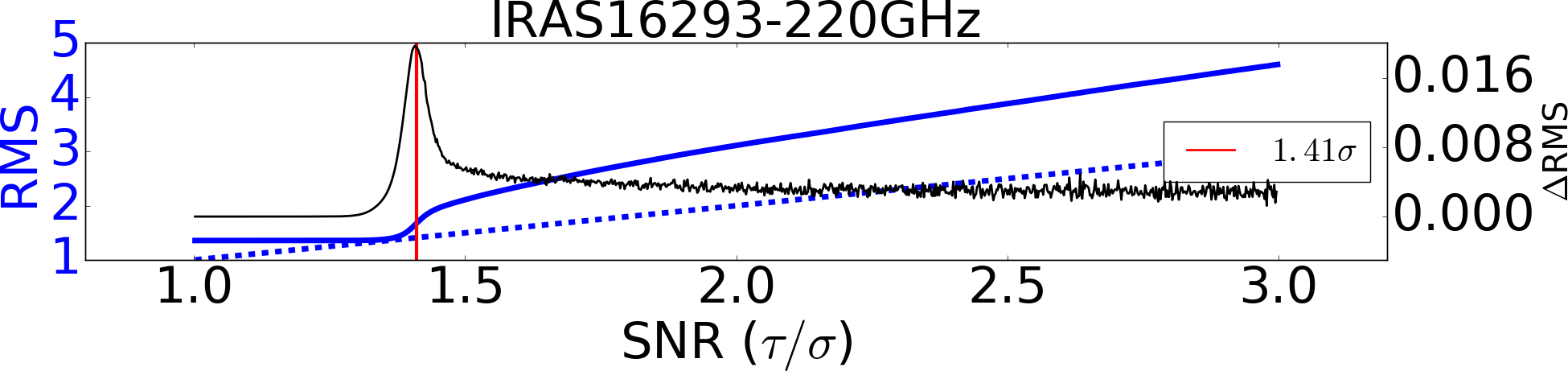}
  \end{center}
  \caption{SNR estimation. The blue lines correspond to the RMS computed after
thresholding the cube using $\tau$. The dotted line is the lowest value that the
RMS can take, the black line shows the slope of the RMS and the red line is the 
SNR selected by the heuristic.}
  \label{fig:rms-snr}
\end{figure}

\begin{figure*}[tbp!]
  \begin{center}
 \Large Orion-KL-CH30H\\
\includegraphics[width=0.90\linewidth]{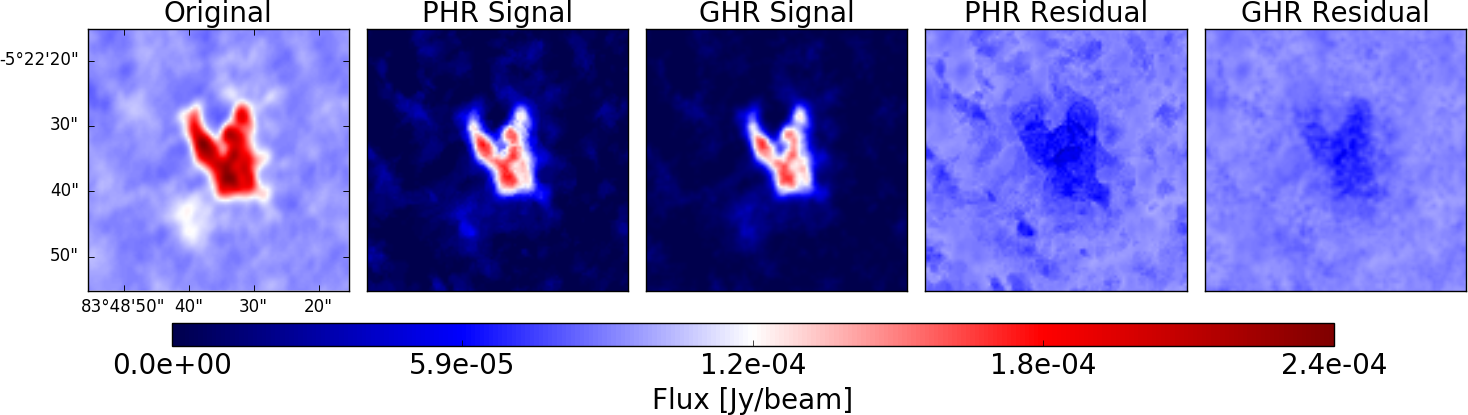}\\
 \Large TW-Hya-CO(3-2)\\
\includegraphics[width=0.90\linewidth]{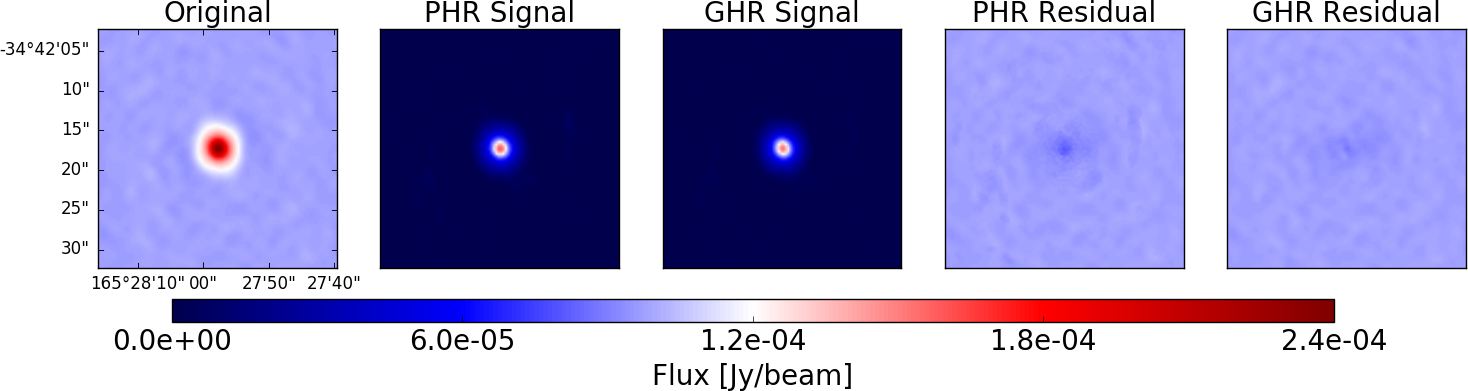}\\
 \Large M100-C0(1-0)\\
\includegraphics[width=0.90\linewidth]{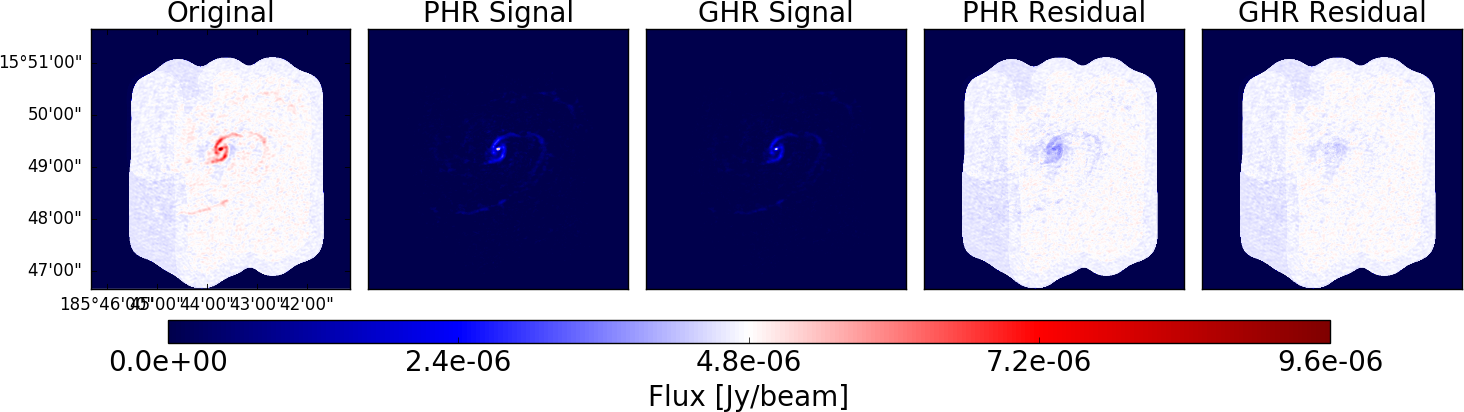}\\
 \Large IRAS16293-220GHz\\
\includegraphics[width=0.90\linewidth]{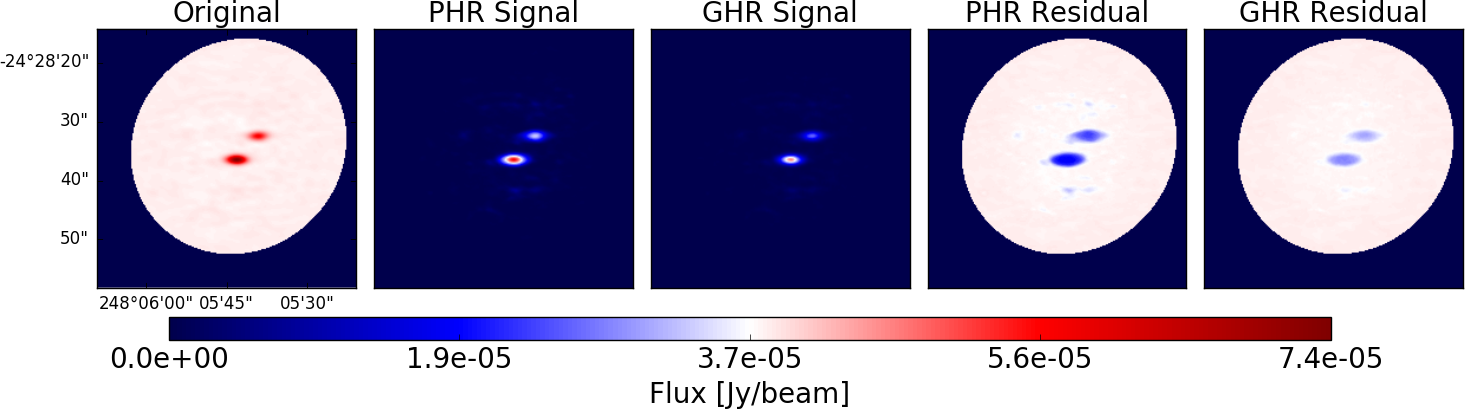}
 \end{center}
  \caption{Noise/Signal Separation using PHR and GHR. For each cube the 
zero-\textit{th} moments for the original data, the projection of the representations
over an empty cube, and the residuals are reported.}
  \label{fig:pghr}
\end{figure*}

\begin{figure}[htbp!]
  \begin{center}
 \large Orion-KL-CH30H\\
\includegraphics[width=0.95\linewidth]{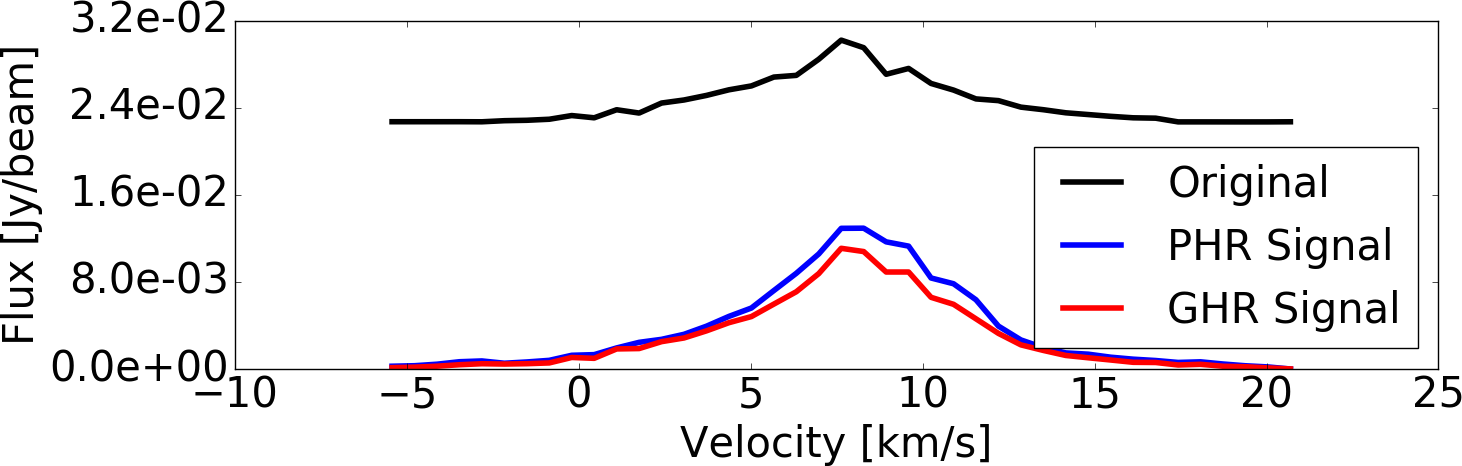}\\
 \large TW-Hya-CO(3-2)\\
\includegraphics[width=0.95\linewidth]{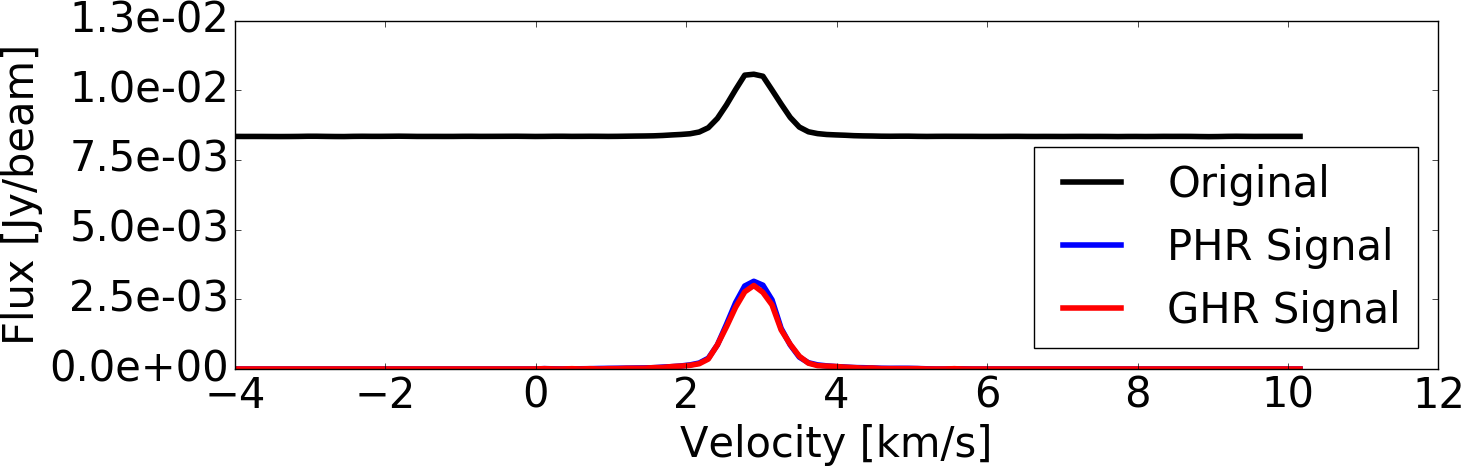}\\
 \large M100-C0(1-0)\\
\includegraphics[width=0.95\linewidth]{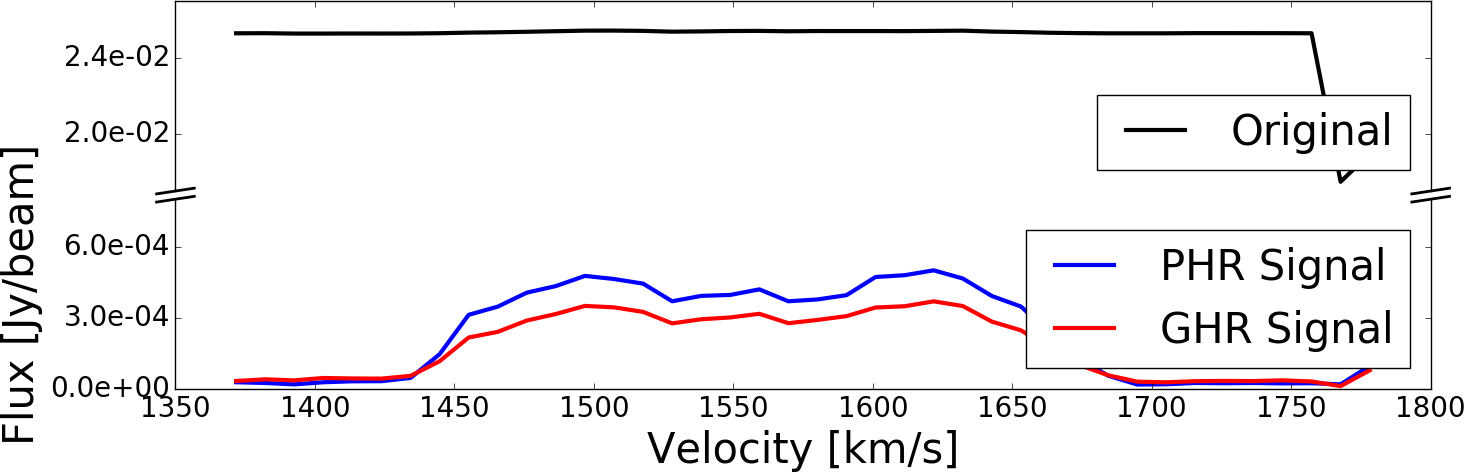}\\
 \large IRAS16293-220GHz\\
\includegraphics[width=0.95\linewidth]{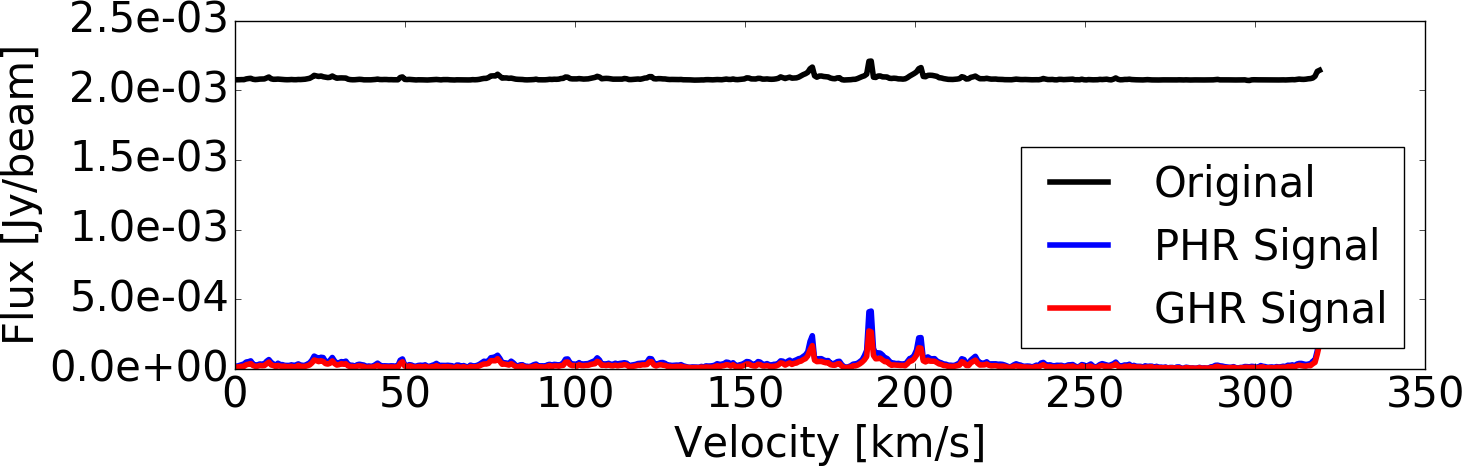}
 \end{center}
  \caption{Spectra of the Noise/Signal Separation using PHR and GHR. The integrated spectra of the same cubes of Figure \ref{fig:pghr} are presented as radial velocities with respect to the reported rest frequency
in the file header.}
  \label{fig:pghrspec}
\end{figure}

To illustrate the benefits of the proposed representation, we consider in our experiments a 
set of spectroscopic cubes from the ALMA Science Verification dataset~\citep{science_verification}. We have chosen these observations because they are thoroughly validated, and obtained from well-known sources for different type of science cases.
In Table~\ref{tab:dataset} the summary of the dimensions and resolutions of the selected cubes are presented. 
The first cube covers a region of $\sim 30 \times 30''$ towards the Kleinmann-Low Nebula (i.e., Orion KL); a massive star-forming region, in which methanol has been observed as a tracer of hot dense gas. Then we have TW Hya; a young star with a transitional disk (in evolution between the states of protoplanetary and debris disks), where a CO detection traces the outer cold disk. The M100 source is a ``grand-design'' spiral galaxy, for which a CO line was used as a tracer of the distribution of molecular hydrogen in the spiral arms. At last, IRAS16293-2422 is a protostar (dense core) that is actually a multiple system for which a cube around 220 GHz was obtained to explore the different lines and chemistry of the envelopes of the components.

For each of these cubes, an automatic transformation procedure was followed.
First, the fluxes of all the valid pixel are standardized in the $[0,1]$ range to ensure numerical stability. Then, the parameter $\sigma$ of the algorithm is estimated
as the RMS of the cube, and an estimation of the signal-to-noise ratio (SNR) is
performed to obtain the $\tau$ parameter\footnote{Even though the algorithm is resistant to aberrant pixels when using spread kernels, the RMS and SNR estimations can be largely affected by outliers. For a robust estimation we recommend a bad pixel treatment of the cube before this step.}. At last, we compute the Pixel-based Homogeneous Representation (PHR) using the identity kernel, and then the Gaussian-based Homogeneous Representation (GHR) using the Gaussian kernel and   Algorithm~\ref{alg:bubble-detect} (\textsc{Bubble-Detect}).

%
%
To tune $\tau$ we propose a heuristic that explores the relation between RMS and the threshold. 
We start at $\tau=\sigma$ and stop at $\tau=3\sigma$, assuming that no flux detection can be done below
$\sigma$, and that everything over $3\sigma$ is a detection. In Figure~\ref{fig:rms-snr} we can see the SNR curve. 
We find the inflection point by looking for the maximum in $\Delta$RMS. 
The key insight of the heuristic is that the change on the RMS slope reflects the
threshold point where the remaining pixels are not longer dominated by noise.
We can notice that for Orion KL the RMS curve decreases monotonously. 
This behavior could be explained by a bad estimation of $\sigma$: the intensity of the
signal in this cube could bias the RMS estimation away from the noise level. 
Regardless, we will still use this value to show that the performance of the algorithm does not fall too much when the noise is overestimated.

\subsection{Pixel-based Homogeneous Representation (PHR)}
\label{sec:phr}

\begin{table*}[tbp]
\normalsize
\centering
\begin{tabular}{l|cccccc}
Name & RMS ($\sigma$) & SNR ($\tau/\sigma$) & PHR & PHR/Bound & PHR/Valid & Time (s) \\
\hline
Orion-KL-CH3OH & 2.60e-06 & 1.010 & 55964 & 0.83 & 0.14 & 0.36 \\
TW-Hya-CO(3-2) & 8.99e-07 & 1.118 & 23796 & 0.92 & 0.02 & 0.11 \\
M100-CO(1-0) & 9.18e-08 & 1.592 & 101890 & 0.98 & 0.01 & 1.02 \\
IRAS16293-220GHz & 5.90e-08 & 1.410 & 328044 & 0.94 & 0.03 & 3.74 \\
\end{tabular}
\caption{PHR Solutions Summary. We report here the number of points used by the representation, including the RMS and the estimated SNR used to obtain it. 
In addition, the proportion with respect to the theoretical bound and 
with respect to the number of valid pixels are reported. 
The last column shows the time spent in computing the representation.}
\label{tab:phr-size}
\end{table*}
\normalsize

\begin{table*}[tbp]
\centering
\begin{tabular}{l|cccc}
Name & GHR & GHR/PHR & GHR/Valid & Time $[s]$ \\
\hline
Orion-KL-CH3OH & 7026 & 0.13 & 0.017 & 37 \\
TW-Hya-CO(3-2) & 1550 & 0.07 & 0.001 & 15 \\
M100-CO(1-0) & 12098 & 0.12& 0.001 & 187\\
IRAS16293-220GHz & 34511 & 0.11& 0.003 & 598 \\
\end{tabular}
\caption{GHR Solutions Summary. In this table we include the
number of elements in the GHR solution, and the ratios with respect to the PHR
solution size and valid pixels of the original cube. We also report the time
spent computing this solution.}
\label{tab:ghr-size}
\end{table*}
\normalsize

Separating signal from noise is the first step in any data analysis task.
Usually, this is done by conducting data thresholding, obtaining a data product of the
same size but with several values in zero. Our proposal uses sparsification to produce
a lighter representation that can be used for data analysis.

Results of the signal-noise separation using our pixel-based
homogeneous representation are reported in Figures~\ref{fig:pghr} and \ref{fig:pghrspec}. We can observe
that the sources have different nature (i.e., star forming region,
protoplanetary disk, spiral galaxy, 
protostar binary system), and for all of them
the representation already captures the spatial structure and the line
spectra of the sources. For instance, M100 shows a clear spiral form in the
zero-\textit{th} moment and a very large FWHM and complex spectrum, while the TW Hya 
shows a very defined protoplanetary disk in the images with a narrow and 
a Gaussian shaped spectral line. Even though some residuals may show significant 
structures in their zeroth moments, please note that they are always below the noise 
level, meaning that no relevant flux is left in the residual.

Table~\ref{tab:phr-size} shows that the number of 
points are near the bound and uses one or two
orders of magnitude less points than the valid pixels of the 
original cube. Please note that this representation is not
a sparse representation of the pixels above a threshold, but a
\textbf{scatter set of independent and identically distributed (i.i.d.) samples} 
of the signal. This means that the representation is not only 
(relatively) compact, but it is also ready for all the 
statistical analysis techniques that rely on this assumption.

A PHR is very fast to compute but it has a major drawback: it does not use a priori information
about the effective resolution of the cube. For instance, in 
ALMA data, a pixel with high intensity that is surrounded only by pixels below 
the noise level it is with high probability an imaging artifact, because the beam size is usually larger than a pixel. Unfortunately, a PHR will consider the point as relevant flux for the
representation.

\subsection{Gaussian-based Homogeneous Representation (GHR)}

Consider now using the Gaussian kernel and the proposed \textsc{Bubble-Detect} method (Algorithm~\ref{alg:bubble-detect}).
In Figures~\ref{fig:pghr} and Figure~\ref{fig:pghrspec} we also report the GHR zeroth moments and spectra in order to compare them with the PHR representation. 
Even though these results seem very similar at first sight, there are several differences 
that we explore in this section.
First, we can notice that GHR images and spectra are much
smoother than for PHR, because the kernel act as a smoothing
operator over the representation, leaving high-frequency noise and artifacts in residuals rather than in signals. 
Also, we can observe that GHR residuals hold a less identifiable structure and less variance than for PHR, resembling more to white noise.

As Table~\ref{tab:ghr-size} shows, GHRs are more compact than PHRs. 
Here we can observe that we reduce to 10\% of the size of PHRs and between 0.1\% and 2\% of the
original valid pixels. However, this is not cost free:
the last column reports the time used to compute the
representation~\footnote{Single-thread execution on a 2.80 GHz i7 processor.}. 
Note that computational times are meant to be spent once, as the representation is data persistent. 
Compared to the theoretical bounds reported in Table~\ref{tab:phr-size},
we can see that the desired SNR is obtained much earlier. 

An important result that we want to highlight is that the size and computational speed of both 
the PHR and GHR representations depend on the complexity of the cube, and not only on the 
resolution or size of the original data. This can be observed by comparing the results for 
TW Hya and Orion KL. Even though the number of pixels in Orion KL is nearly a third of the pixels 
in TW Hya, the complex structure of the star forming region requires 4.5 times more GHR elements than the 
protoplanetary disk (3.6 times for PHR), with a similar proportion for the computational time. 

\section{Applications of GHR}
\label{sec:apps}

\begin{figure}[htbp]
  \begin{center}
\includegraphics[width=0.80\linewidth]{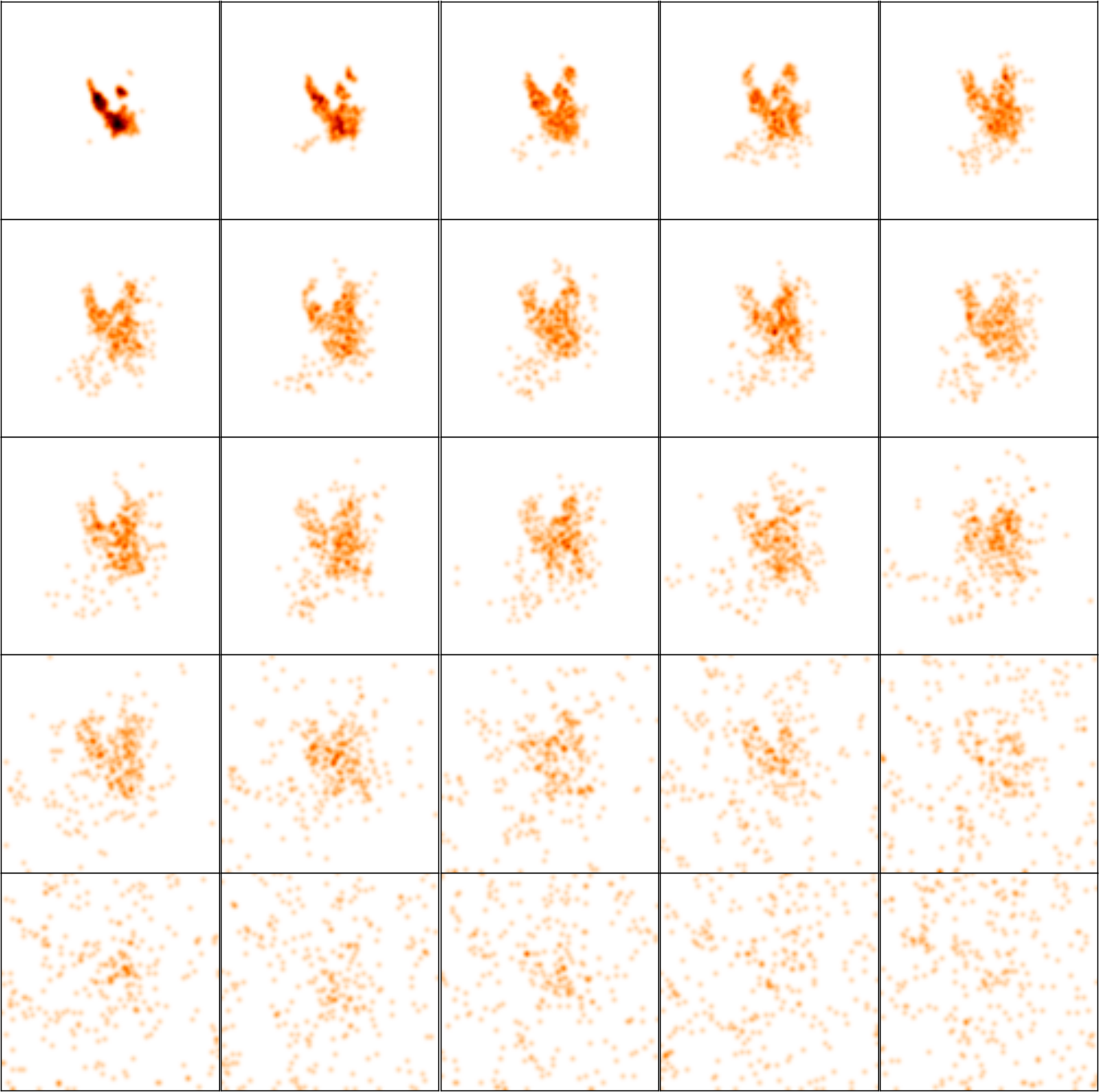}\\
 \end{center}
  \caption{Vertical decomposition of Orion-KL's GHR in 25 levels, each one with the
same total flux.}
  \label{fig:vert-thr}
\end{figure}

\begin{figure}[htbp]
\begin{center}
\includegraphics[width=0.8\linewidth,valign=t]{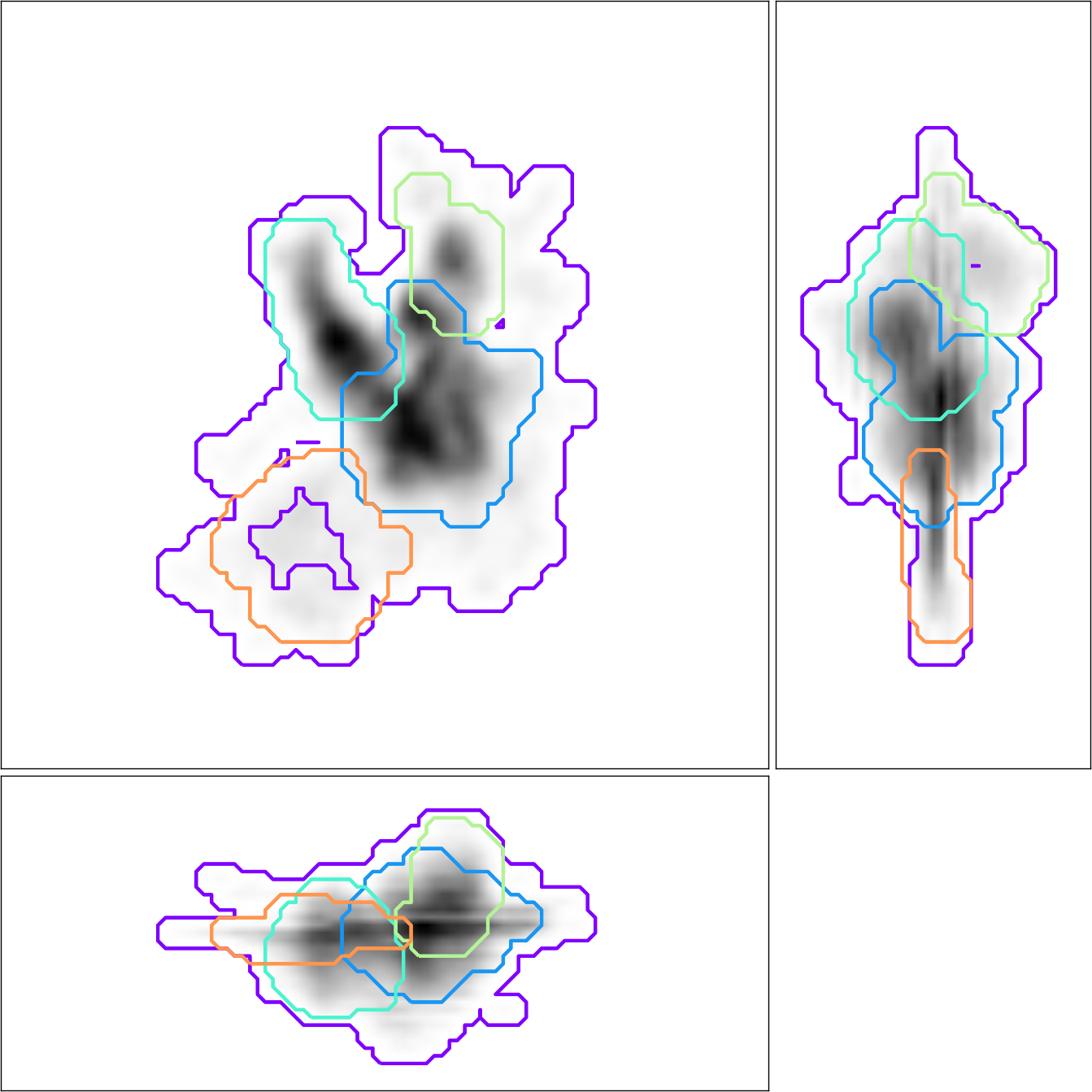}\\
\end{center}
\caption{Detected Clumps in Orion KL using Spectral Clustering. The image shows the contours of the 3D clumps projected over three views: stacked frequency, stacked right-ascension and stacked declination).}
\label{fig:orion-clump}
\end{figure}

In this section we use the GHR introduced in Section~\ref{sec:hcr}
to give a few examples of its possible data analysis usages. 

To recapitulate, the main advantage of GHR with respect to PHR is that
the cube is now represented by a set of elements that
can be seen as independent and identically distributed (i.i.d) samples of the 
underlying signal. Consequently, a vast number of statistical and machine learning 
techniques found in packages like \texttt{scikit-learn} \citep{scikit-learn} can be directly
applied to the representation. 

\subsection{Vertical Thresholding and Clumping}

Let us consider the Orion KL data presented in Section~\ref{sec:hcr}.
In pixel-based representations, thresholding implies removing all
the flux of some pixels and maintain the complete flux of the rest of them.
Our representation allows performing \emph{vertical thresholding}, because
the samples are produced in flux decreasing order. Figure~\ref{fig:vert-thr} shows
a decomposition of the data in 25 levels, each one with the same total
flux. Please note that this is done over 3D representations and then projected
to 2D images. This allow us to visualize the profile of the emission,
and select data vertically. For example, we selected the first 10 levels for the 
next step, not only performing denoising as in traditional thresholding, but 
also reducing the size of the representation.


Detecting clumps is an important application that is strongly simplified by
our representation, because it reduces to a clustering problem. 
For the Orion KL data we have selected \emph{spectral clustering}~\citep{Shi00}, a well-known technique for image segmentation. 
We conducted spectral clustering using an arbitrary number of clusters equal to 5.  Figure~\ref{fig:orion-clump} show clumps projections to 2D views. 
For example, we can observe the detection of components that differ from the core spatially and due to thinner broadening (orange), red/blue-shifted components (yellow/green), and extended ones with a few peaks that suggest more clusters are needed. This paper does not address the problem of finding the number of clusters automatically, but provides a light-wighted representation that can be used for re-clustering at will for manual or automatic parameter tunning.

\subsection{Source Selection and Manifold Representation}

\begin{figure}[tbp]
\begin{center}
\includegraphics[width=0.8\linewidth,valign=t]{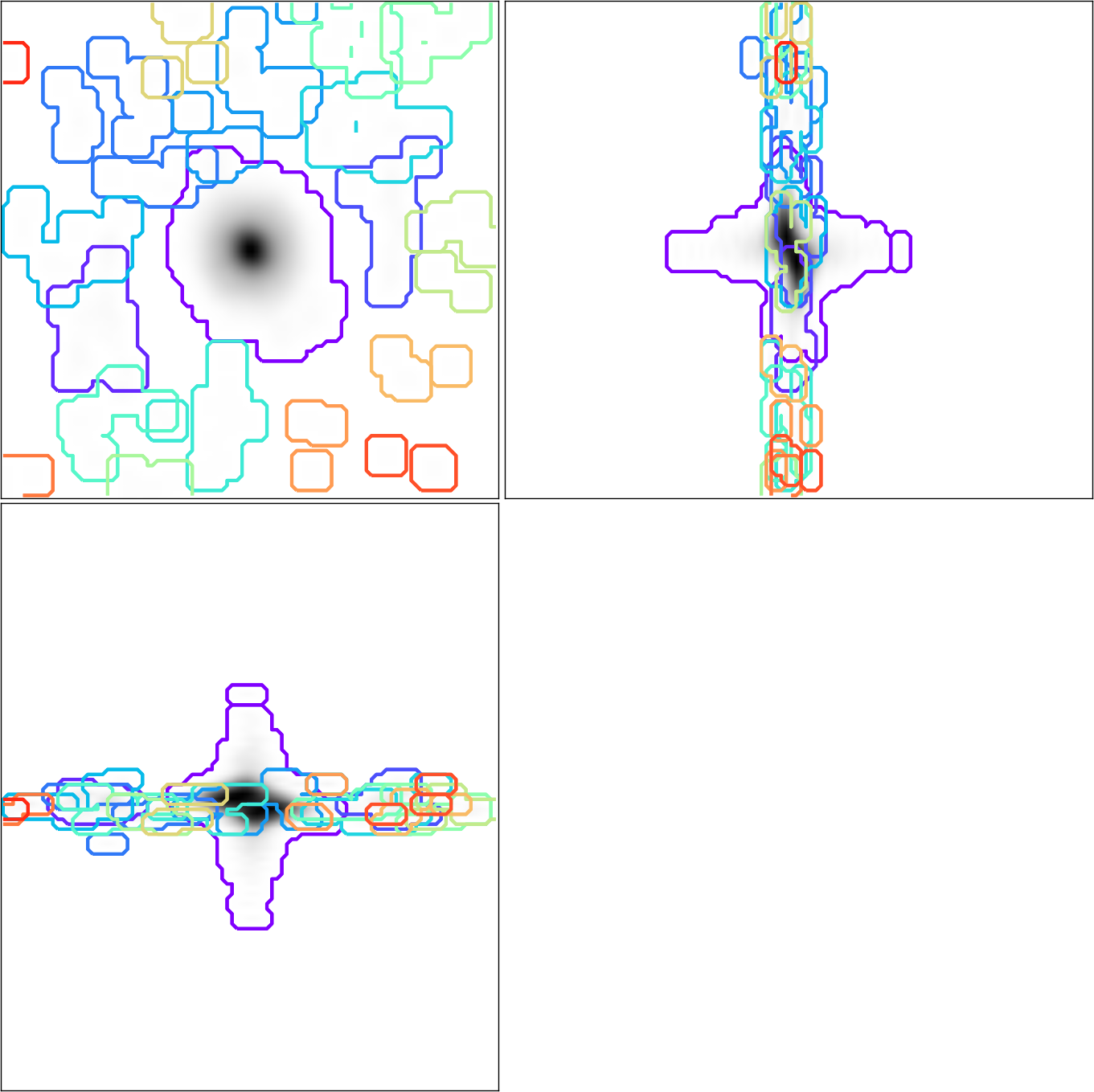}\\
\end{center}
  \caption{Detected clumps in TWHya using Mean Shift. The image shows the contours of the 3D clumps
projected over three views (stacked frequency, right-ascension and declination).}
  \label{fig:hydra-clump}
\end{figure}

\begin{figure}[tbp]
\begin{center}
\includegraphics[width=0.8\linewidth,valign=t]{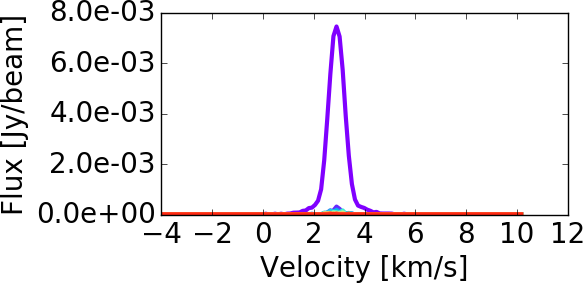}
\caption{Spectra of Clumps in TW Hya. The colors of the clumps coincides with Figure~\ref{fig:hydra-clump}.}
\label{fig:hydra-clump-spectra}
\end{center}
\end{figure}

\begin{figure}[tbp]
  \begin{center}
\includegraphics[width=0.45\linewidth,valign=t]{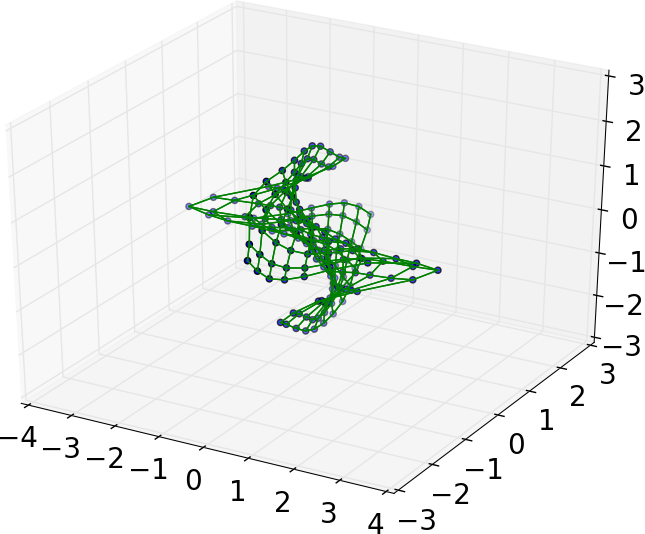}
\includegraphics[width=0.45\linewidth,valign=t]{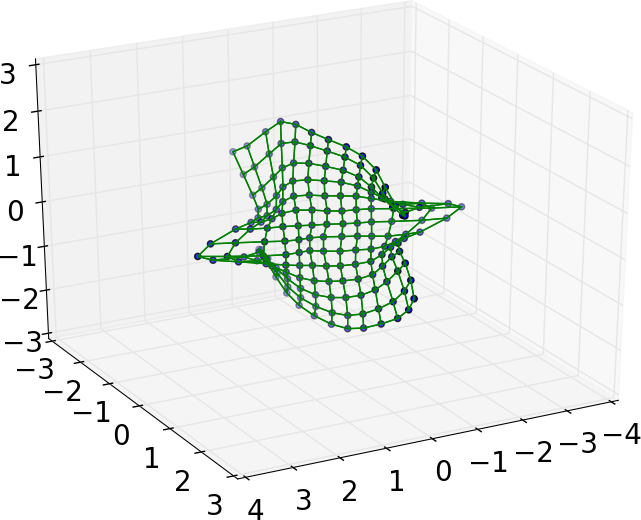}\\
\includegraphics[width=0.45\linewidth,valign=t]{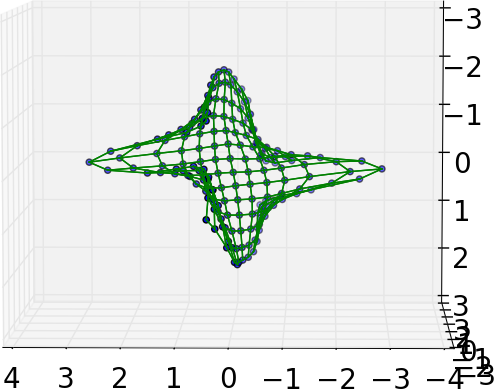}
\includegraphics[width=0.45\linewidth,valign=t]{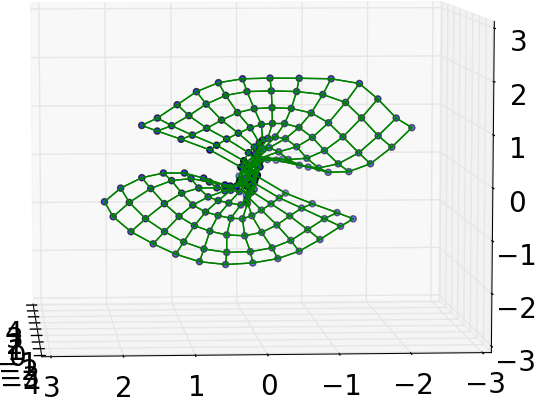}\\
 \end{center}
  \caption{Synthesized manifold of TW Hya, including isometric (top) and side
viewpoints (bottom).}
  \label{fig:twhydra-som}
\end{figure}

Let's consider now the TW Hya data. Even though this might look as a unique source,
the data contains other non-centered emissions. Therefore, we can use a non-parametric
clustering algorithm called \emph{mean shift} \citep{Cheng95} that finds local maxima and
construct clusters around them. The algorithm automatically segments the data
in 19 clusters showed in Figure \ref{fig:hydra-clump}. We select only the highest emission
cluster (magenta) for the next step, because the individual fluxes of all the other clumps are marginal compared to that one. In fact, we can see in Figure~\ref{fig:hydra-clump-spectra} that the other clusters have a very small and indistinguishable flux spectra compared to the magenta line.

The selected cluster have a clear shape both in spatial and spectral dimensions. 
Accordingly, we can synthesize a manifold to represent it. To do this we use an
artificial neural network technique called \emph{self organizing maps} (SOM), 
that represents the data by a bi-dimensional grid of connected
neurons~\footnote{We use the SOMPY package that can be found in github:
https://github.com/sevamoo/SOMPY.} \citep{kohonen82}. 
In Figure~\ref{fig:twhydra-som} the manifold representation of the source is shown, starting from a planar grid. This representation might help, for example, to better understand the topology of the source when comparing with theoretical models. 

\subsection{Filtering and Reduction by Clustering}

\begin{figure*}[tbp]
  \begin{center}
\includegraphics[width=0.27\linewidth]{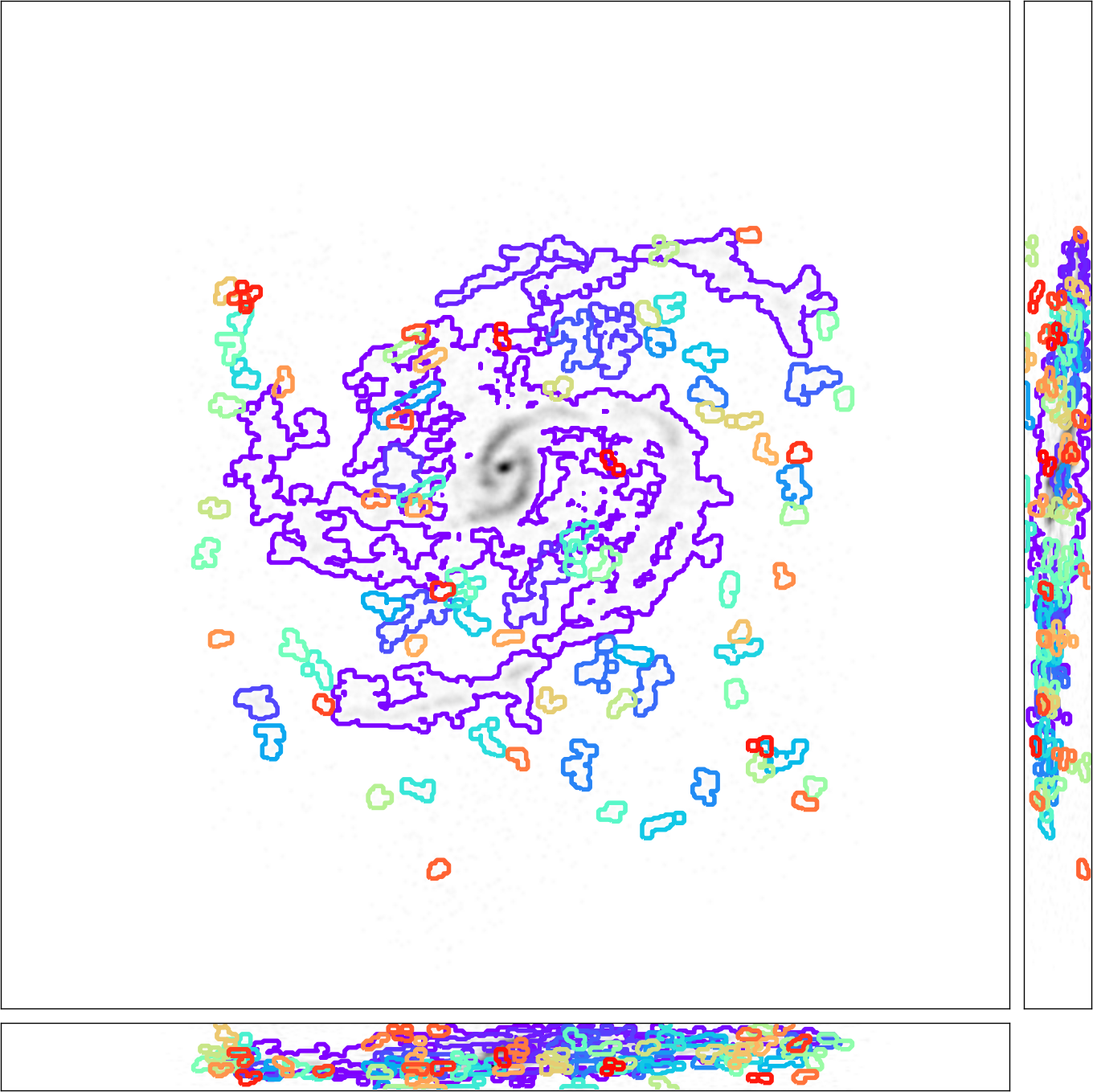}
\includegraphics[width=0.27\linewidth]{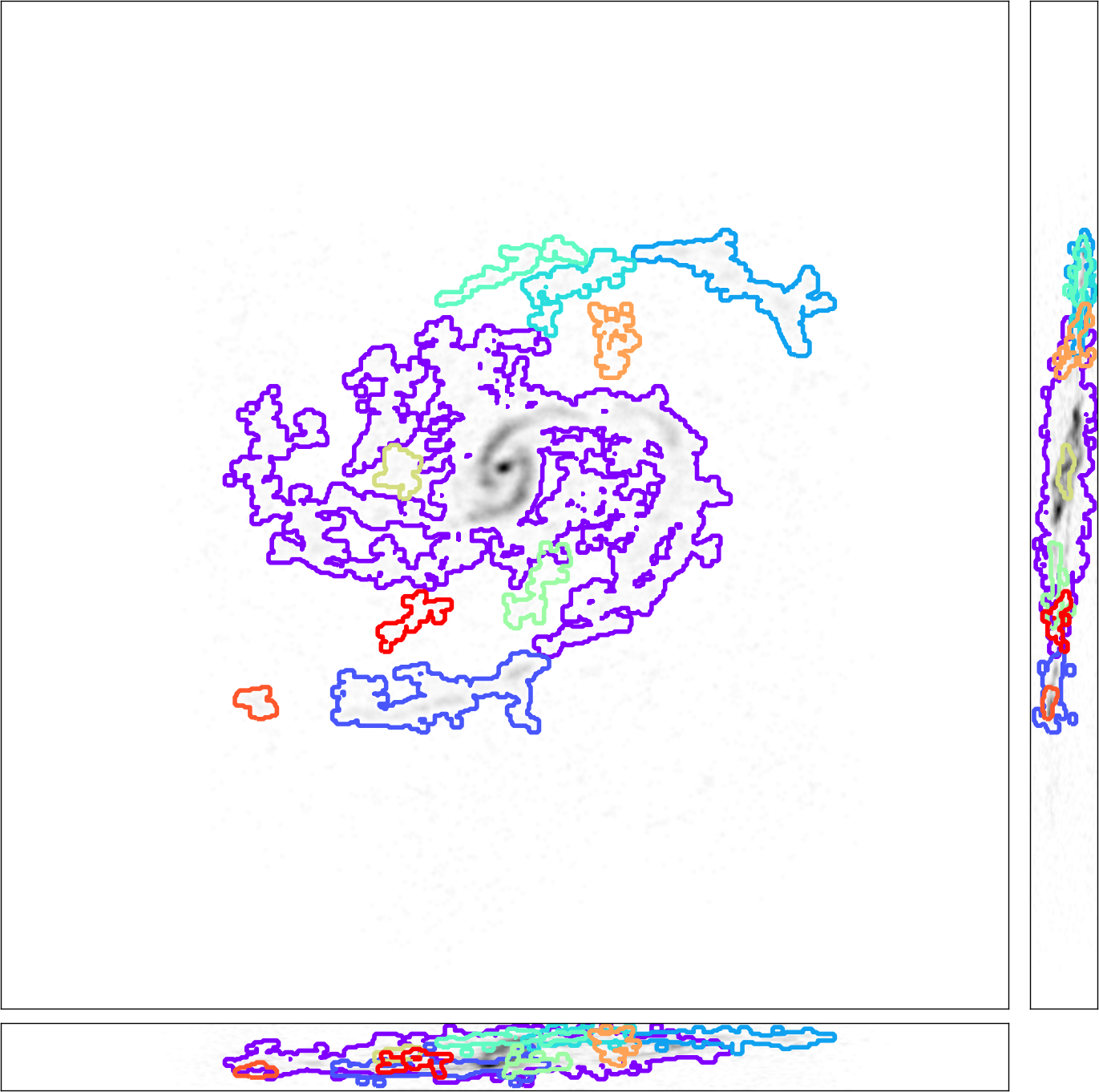}
\includegraphics[width=0.27\linewidth]{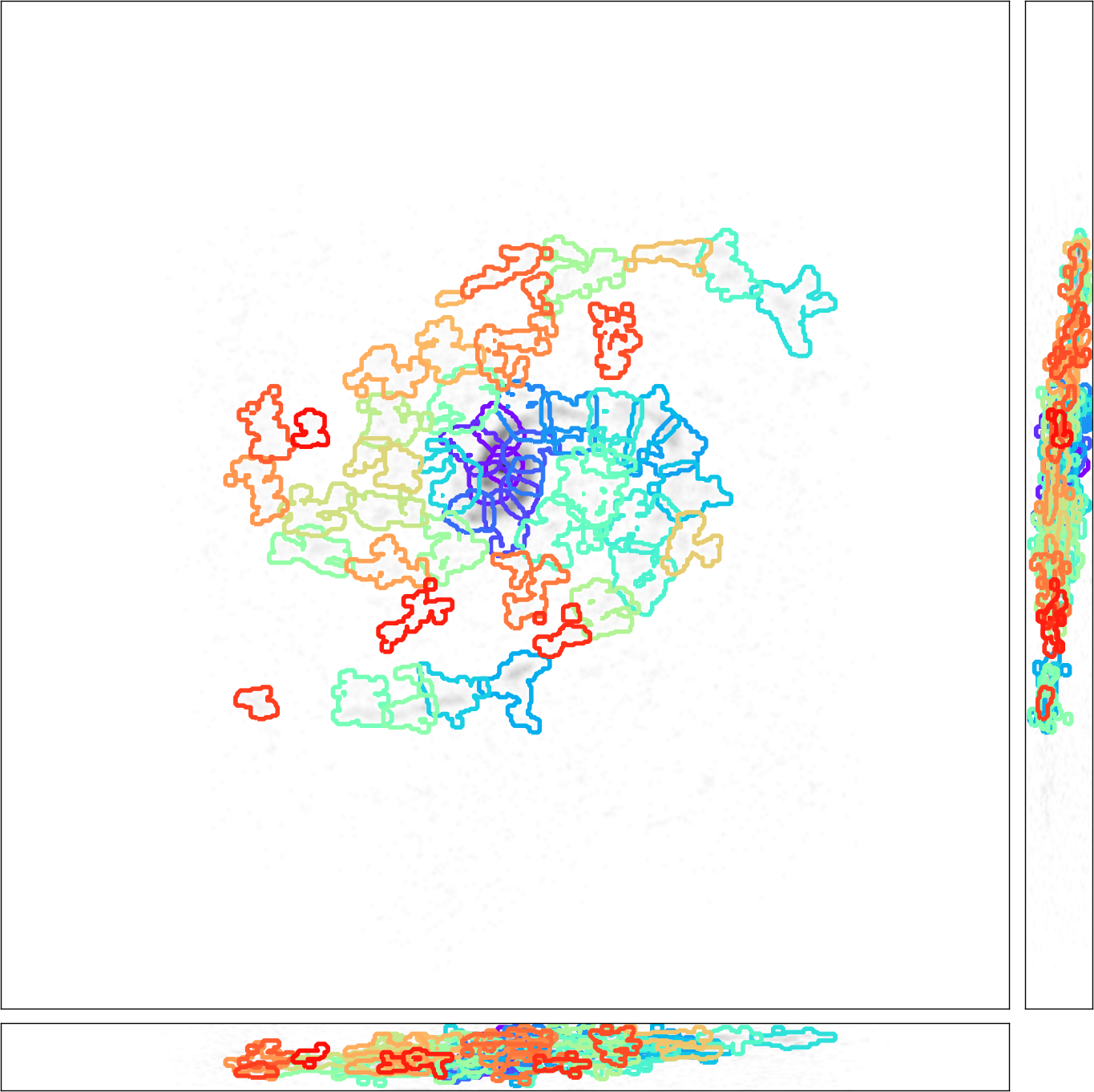}\\
\end{center}
  \caption{Clustering stages for M100. The left image 
  shows the projections of the clusters obtained by DBSCAN. Similarly, 
  the center image shows only the 10 brightest clusters, and the right 
  image shows the 50 clusters obtained by K-Means from the filtered data.}
  \label{fig:m100-clust}
\end{figure*}

\begin{figure}[tbp]
  \begin{center}
\includegraphics[width=0.85\linewidth]{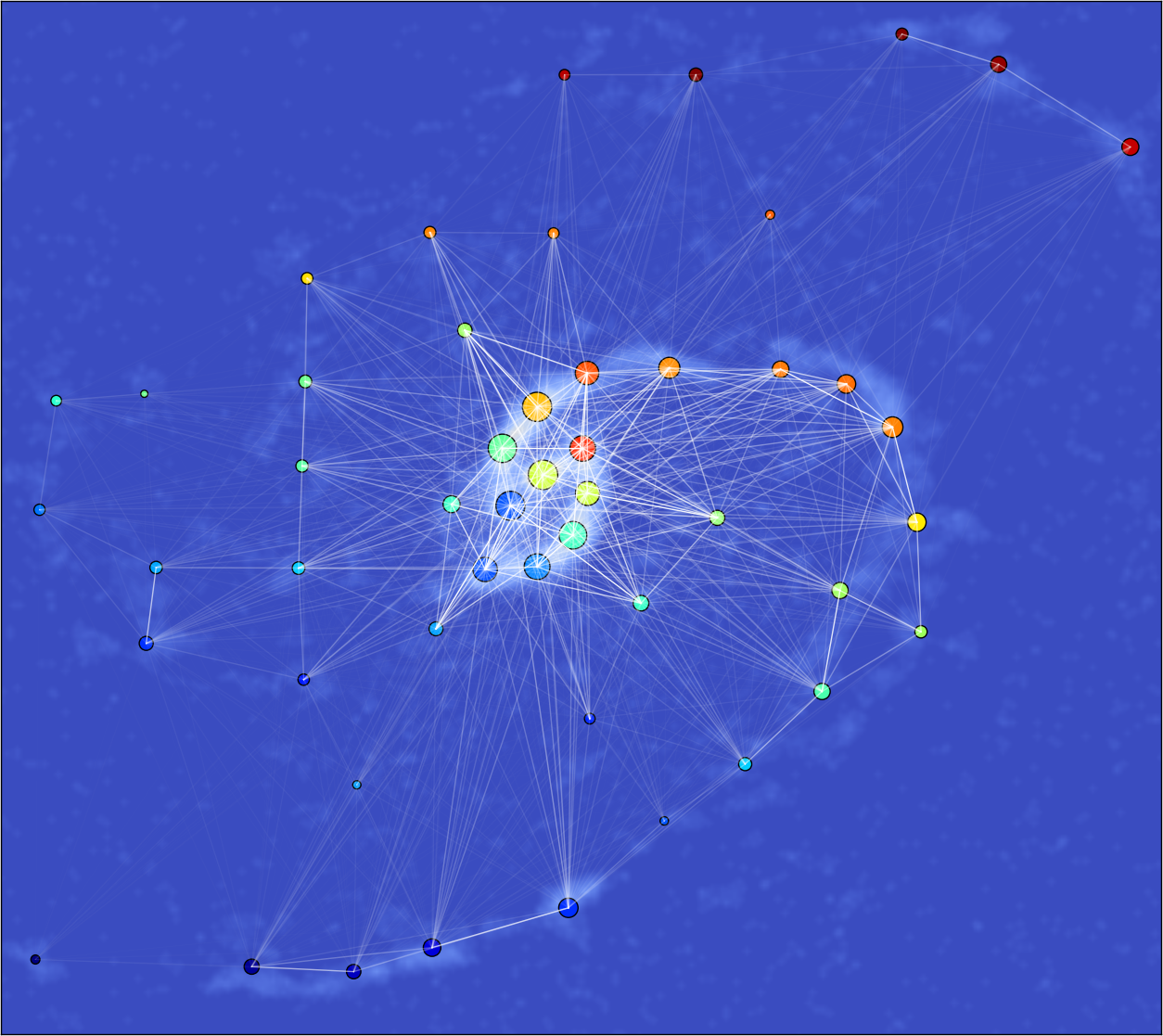}
\end{center}
  \caption{Schematic Graph of M100.  
  This figure shows the projected schematic graph representation of the galaxy,
  plotted over a heat-map version of the zero-\textit{th} moment.}
  \label{fig:m100-clust2}
\end{figure}

Clustering algorithms are not only useful to detect clumps or sources, but can
work as tools for filtering and summarizing data. For the M100 data, we 
use the non-parametric DBSCAN algorithm \citep{Ester96} to form clusters of neighboring points 
while discarding the unconnected ones ($\epsilon=7 [pix]$). In a second stage
we selected only the 10 brightest clusters to discard very small clusters.
At last, we use K-Means to summarize the data in
only 50 points. This process is shown in Figure~\ref{fig:m100-clust}, 
and the final result is reported in the schematic representation of Figure~\ref{fig:m100-clust2}.  In this representation the size of each node represents the flux, the color of each
node represents the radial velocity, and the edges correspond to a 
``gravity'' connection (i.e., $\propto  \frac{F_i F_j}{d^2}$) between nodes. The spiral structure and the velocity gradients of M100 can be easily observed in the figure. For example, this representation might help to seed a n-body simulation starting from real data.



\subsection{Gaussian Clumps and Parameter Estimation}
\begin{figure*}[tbp]
  \begin{center}
\includegraphics[width=0.45\linewidth]{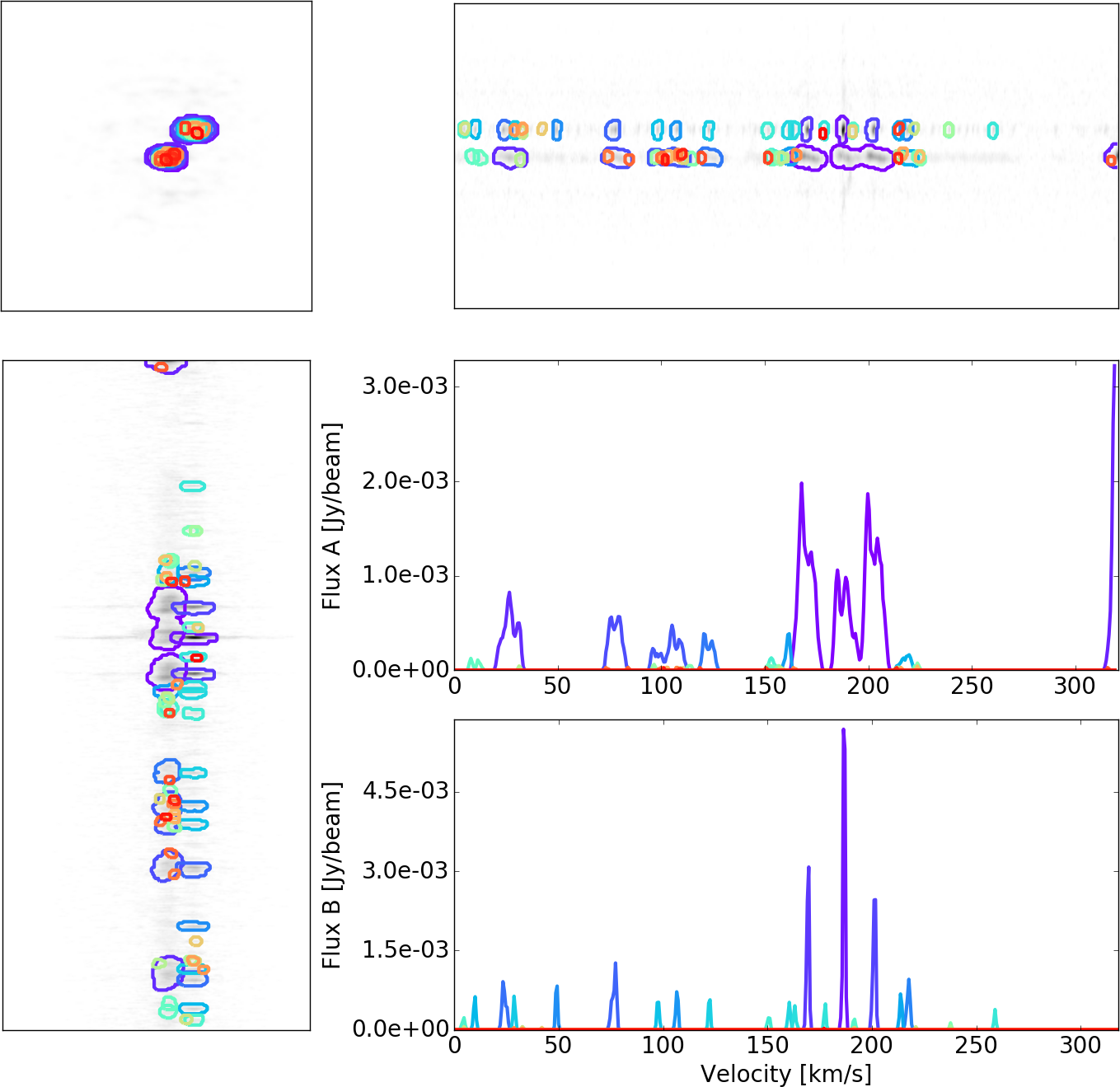}\hfill
\includegraphics[width=0.45\linewidth]{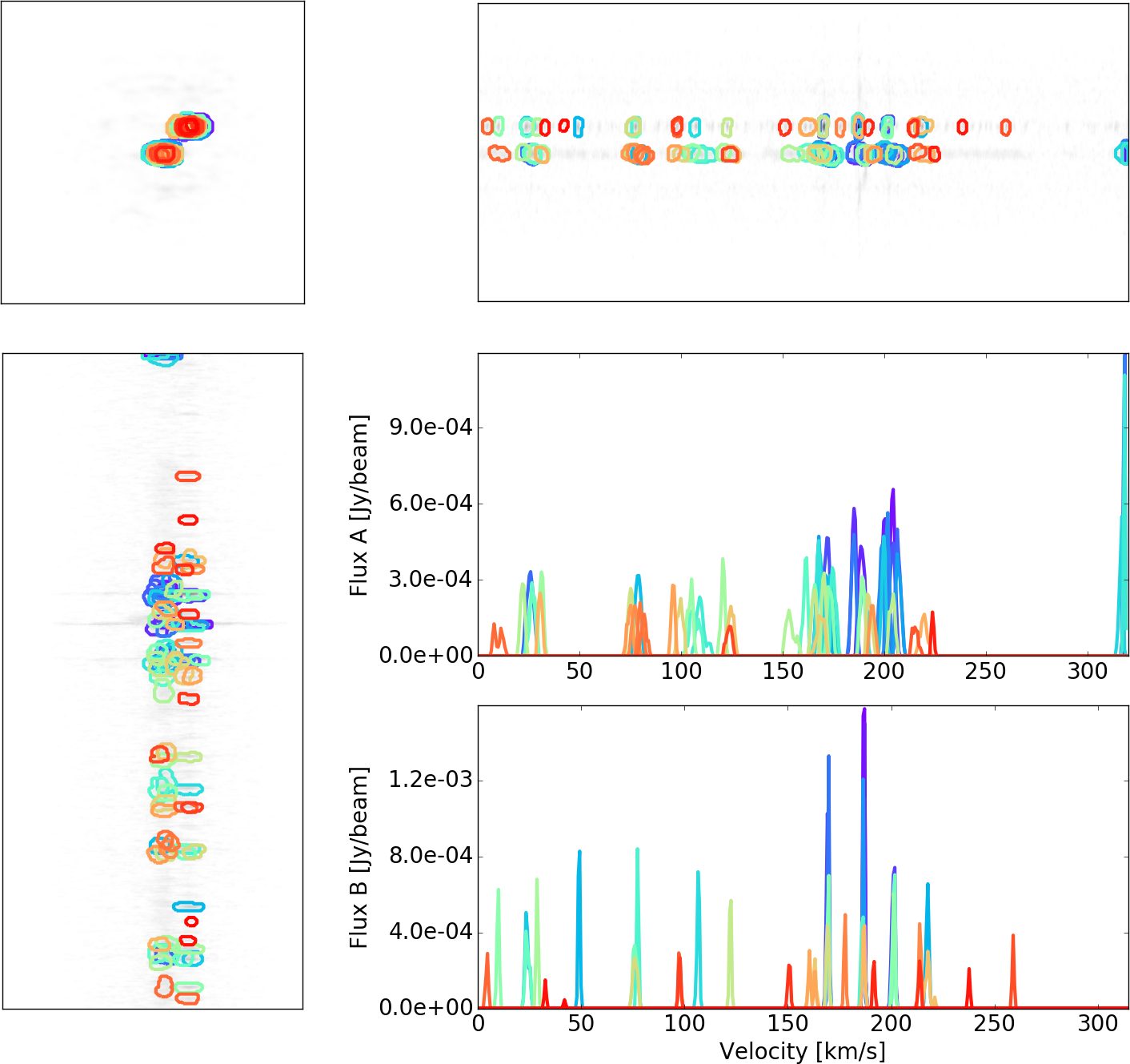}
 \end{center}
  \caption{IRAS16293-220GHz clustering results. The left figure shows the clustering result using DBSCAN and the right figure shows the clustering result using K-Means. The spectral lines are shown for sources A and B of the system separately.}
  \label{fig:iras-clust}
\end{figure*}

A key difference of the GaussClump algorithm \citep{stutzki90} compared to the pixel-based
ones is that each Gaussian clump is represented by a very compact set of parameters with 
astrophysical meaning. 
The main problem is that clumps are usually more complex
structures than Gaussians. Even though they can be efficiently represented
by a mixture, analyzing a complexly blended structure of several components can be a hard task.

However, there are some cases where independent clumps can be accurately
represented by Gaussians, like in the IRAS16293 data presented in Section~\ref{sec:emp}.
For this data we first run a DBSCAN with $\epsilon = 1.5 [pix]$ because each
component is very compact, discarding noise or very low-energy lines
in the data, as can be seen in Figure~\ref{fig:iras-clust}. Then we
use the Expectation Maximization algorithm \citep{Dempster77} to fit a Gaussian Mixture 
of 100 components. The results are shown in Figure~\ref{fig:iras-clust}.

Please note that the obtained Gaussian mixture is in the space of our homogeneous bubbles, and not in the pixel domain. 
However, as every bubble is a Gaussian function, we can use analytic expressions to compute the moment preserving parameters of each Gaussian cluster \citep{crouse11}.
Let us denote the position vector of the $i$-th bubble in the cluster as
$\mu_i$, $\Sigma_0$ be the bubble shape and $c$ be its integrated energy.
Then, the moments of the cluster $\Omega$ can be computed as indicated in Equation~\ref{eq:WCS}.

\begin{align*}
c_{\Omega} & = \sum_i c = c|\Omega| \\
\mu_\Omega & = \frac{\sum_i \mu_i}{|\Omega|}
\end{align*}

\begin{equation}
\Sigma_\Omega = \Sigma_0 + \frac{\sum_i \mu_i \mu_i^\top}{|\Omega|}
\label{eq:WCS}
\end{equation}

Now we can compute the parameters used in GaussClumps by solving a determined system of linear equations.
The parameters are clump intensity (N), center position (RA, DEC and FREQ), spatial FWHM semiaxes (SA 1 and 2), 
spatial orientation (ANGLE), spectral FWHM, and the gradients $\partial
RA/\partial FREQ$ (GRA) and $\partial DEC/\partial FREQ$ (GDEC). 
As an example of these results, we report all the clumps (lines)
with a spectral FWHM larger than $3 [pix]$ in Table~\ref{tab:cres}.

\begin{table*}[tbp]
\normalsize
\centering
\begin{tabular}{l|cccccccccc}
ID & $|\Omega|$ & $\mu_{RA}$ & $\mu_{DEC}$ & $\mu_{FREQ}$ & SA1 & SA2 & ANGLE & FWHM &
GRA & GDEC\\
\hline
3 & 408 & 121.22 & 108.51 & 39.71 & 2.98 &
3.58 & 30.81 & 3.37 & -0.17 & -0.09 
\\
33 & 291 & 135.66 & 127.66 & 43.91 & 2.98&
5.26 & 0.13& 3.54& -0.09 &
-0.01 \\
66 & 291 & 112.84 & 106.62 & 258.88 & 2.75 &
3.36 & 51.16 & 3.03 & -0.15 & 0.16 
\\
87 & 252 & 116.95 & 108.46 & 290.49 & 2.79 &
3.63 & -20.11 & 3.09 & -0.12 &
-0.02 \\
51 & 222 & 117.49 & 108.27 & 231.05 & 2.61 &
3.49 & -8.95 & 4.37 & -0.18 &
-0.07 \\
44 & 192 & 114.39 & 108.20 & 158.50 & 2.74 & 3.33 &
7.45 & 3.05 & -0.05 & 0.09 \\
67 & 189 & 132.94 & 127.89 & 243.11 & 2.98 &
3.41 & -21.98 & 3.16 & 0.10 & 0.25 
\\
31 & 177 & 118.20 & 107.76 & 329.36 & 2.61 &
3.34 & 16.79 & 3.38 & -0.04 & 0.18
\\
26 & 135 & 118.40 & 108.44 & 15.07 & 2.70 &
2.97 & 6.32 & 4.08 & -0.41 & 0.29 
\\
\end{tabular}
\caption{Detected clumps with a spectral FWHM larger than $3[pix]$. The parameters of each Gaussian clump follows the \citet{stutzki90} model. For simplicity we have not
converted these results to WCS, so they are expressed in pixels.}
\label{tab:cres}
\end{table*}

\section{Conclusions}
\label{sec:conc}

We introduced a representation based on homogeneous volumes that enables data analysis techniques 
to be directly applied to spectroscopic cubes. The algorithm 
obtains this representation by subtracting homogeneous volumes from the original cube. 
Our experiments show that the proposal has a good behavior in terms of signal and noise separation. Even 
though we show examples only for 3D data, the technique could be used for any dimensionality 
by configuring an appropriate kernel. The algorithm proves to have some basic properties that can help to be accepted 
by the astronomical community (i.e., determinism, bounded iterations, positivity constrain, etc.) and 
produces a compact representation of the data ($\leq 1\%$) in a reasonable time, while maintaining 
enough information to perform analysis. Computational time and compactness depend on the target SNR that want to 
be achieved. Thus, we presented an heuristic to estimate this parameter.

In addition, we presented a few examples of usages of the representation, trying to cover different 
astronomical objects and techniques, showing the versatility of the proposed representation. These 
examples did take at most a few seconds of computation time in a general-purpose computer due to the compactness of the 
representation. 


\subsection{Future Work}

This work opens several research directions both in the theoretical and the practical aspects.
First of all, our results show that the presented bound is clearly very loose for the 
Gaussian kernel. We believe that a tighter bound can be found specifically for GHR. 
Formally, we have only 
shown that the desired SNR is reached, but no statistics on the residual cube are given.
We believe that the properties of the algorithm can be explored more deeply if information-theory 
measures (e.g. mutual information) are computed for the residual, which will require to produce 
reasonable synthetic data to compare to a ground-truth reference. Also from the theoretical point of 
view, the forced compact support of the Gaussian function will produce an accumulative error when the 
bubbles are used for estimation. Here, the compact support could be analytically forced by convolving 
the Gaussian kernel with an appropriate smooth step function in all directions. Another interesting 
follow up could be to study and improve the SNR estimation heuristic by using the current advances on 
background estimation. From a more practical 
point of view, the current article does not address the problem of propagating uncertainties. The measurements are 
usually accompanied by their uncertainties, and each data reduction process, such as the GHR 
representation or the Gaussian Mixture fitting, incorporates more uncertainty to the final data. 
This is a very interesting research line that it must be studied if our representation gains popularity. 

A more specific yet very interesting improvement can be made for interferometric data: the compact representation 
could be combined with the image synthesis procedure to reduce error propagation. Interferometric data cubes 
are synthetic data products obtained by an incomplete inverse transformation from a Fourier plane (visibilities). 
More precisely, the image synthesis  \citep{thompson08} is usually done by interpolating the non-uniform visibilities coordinates into a grid where the fast Fourier transform can be applied. Then the cube is improved by an iterative process called CLEAN 
\citep{hogbom74} that selects the brightest pixel, conducts a convolution with a unique functional component, 
and finally it subtracts it from the residual. This is very similar to the process we apply for obtaining homogeneous 
representations, so we suspect our algorithm can be merged with CLEAN to obtain a method that preserves 
the simplicity of CLEAN, adds the positivity constraint of our method, and produce a compact 
representation that can be used for analysis with less accumulative error.

We strongly believe that our representation can be used for boosting astronomical 
research, so we plan to use it in real-world science cases that needs advanced statistical analysis, 
such as emissions with very low SNR, too many or blended spectral lines, too many dimensions, very large 
mosaics, etc. Also, we plan to address the \emph{content-aware} data discovery problem: even though 
there are plenty of services that provide data discovery nowadays (e.g. VO services), almost all of 
them are based in the annotated metadata rather than in the image content. Our representation could 
provide a compact representation that is fast to analyze as we show in this article, allowing the 
astronomer to search for data that fulfill more complex properties than the ones declared in the 
metadata. For this application in particular, the kernel size could be increased producing compact 
representations with less precise information in order to cope with a large number of files.  

%
\section{Acknowledgements}
\label{sec:ack}
This paper makes use of the following ALMA data: 
\begin{itemize}
\item ADS/JAO.ALMA\#2011.0.00001.SV
\item ADS/JAO.ALMA\#2011.0.00004.SV 
\item ADS/JAO.ALMA\#2011.0.00007.SV 
\item ADS/JAO.ALMA\#2011.0.00009.SV.
\end{itemize}
ALMA is a partnership of ESO (representing its member states), NSF (USA) and NINS (Japan), together with NRC (Canada), NSC and ASIAA (Taiwan), and KASI (Republic of Korea), in cooperation with the Republic of Chile. The Joint ALMA Observatory is operated by ESO, AUI/NRAO and NAOJ.
This work has been partially funded by CONICYT PIA/Basal FB0821,
CONICYT PIA/Basal FB0008,
and FONDEF IT 15I10041. \newline

%
\section*{References}

\bibliographystyle{elsarticle-harv} 
\bibliography{bubble}
%
%
%
%
%
%


\end{document}